\documentclass[12pt]{article}

\usepackage{amsmath}
\usepackage{amssymb}
\usepackage{amsthm}
\usepackage[pdfpagemode=None]{hyperref}
\usepackage{graphicx}
\usepackage{pstricks}
\usepackage{psfrag}

\newtheorem{thm}{Theorem}[section]

\newtheorem{lem}[thm]{Lemma}

\newtheorem{assumption}[thm]{Assumption}

\newtheorem{pr}[thm]{Proposition}

\newtheorem{definition}[thm]{Definition}

\newtheorem{example}[thm]{Example}
\newenvironment{exmp}{\begin{example}\rm}{\end{example}}
\newtheorem{remark}[thm]{Remark}
\newenvironment{rem}{\begin{remark}\rm}{\end{remark}}
\newtheorem{tab}{Table}

\newcommand{\eps}{\varepsilon}

\title{Menger's Paths with Minimum Mergings~\footnote{This work is partially supported by a grant from the University Grants Committee of the Hong Kong Special Administrative Region, China (Project No. AoE/E-02/08).}}

\author{Guangyue Han\\
  {\normalsize Department of Mathematics}\vspace{-1mm} \\
  {\normalsize University of Hong Kong}\vspace{-1mm} \\
  {\normalsize Pokfulam Road, Hong Kong}\\
  {\normalsize {\em e-mail:\/} ghan@hku.hk}}

\date{{\normalsize \today}}

\begin{document}\maketitle\thispagestyle{empty}

\begin{abstract}
For an acyclic directed graph with multiple sources and multiple sinks, we prove that one can choose the Menger's paths between the sources and the sinks such that the number of mergings between these paths is upper bounded by a constant depending only on the min-cuts between the sources and the sinks, regardless of the size and topology of the graph. We also give bounds on the minimum number of mergings between these paths, and discuss how it depends on the min-cuts.
\end{abstract}

\section{Introduction}

Let $G(V, E)$ denote an acyclic directed graph, where $V$ denotes the set of all the vertices (points) in $G$ and $E$ denotes the set of all the edges in $G$. Using these notations, the edge-connectivity version of Menger's theorem~\cite{Menger1927} states:
\begin{thm}[Menger, 1927]
For any $u, v \in V$, the maximum number of pairwise edge-disjoint directed paths from $u$ to $v$ in $G$ equals the
min-cut between $u$ and $v$, namely the minimum number of edges in $E$ whose deletion destroys all directed paths from $u$ to $v$.
\end{thm}
\noindent We call any set consisting of the maximum number of pairwise edge-disjoint directed paths from $u$ to $v$ a set of {\em Menger's paths} from $u$ and $v$. Apparently, for fixed $u, v \in V$, there may exist multiple sets of Menger's paths.

For $m$ paths $\beta_1, \beta_2, \cdots, \beta_m$ in $G(V, E)$, we say these paths {\it merge} at $e \in E$ if
\begin{enumerate}
\item $e \in \cap_{i=1}^m \beta_i$;
\item there are at least two distinct $f, g \in E$ such that $f, g$ are immediately ahead of $e$ on some $\beta_i, \beta_j$, $j \neq i$, respectively.
\end{enumerate}
Roughly speaking, condition $1$ says that $\beta_1, \beta_2, \cdots, \beta_m$ {\em internally intersect} at $e$ (namely, all $\beta_i$'s share a common edge $e$), condition $2$ says immediately before all $\beta_i$'s internally intersect at $e$, at least two of them are different. We call $e$ together with all the subsequent edges shared by all $\beta_i$'s (until they branch off) {\it merged subpath} by $\beta_i$ ($i=1, 2, \cdots, m$) at $e$; and we often say all $\beta_i$'s merge at the above-mentioned merged subpath. {\bf In this paper we will count number of mergings without multiplicities: all the mergings at the same edge $e$ will be counted as one merging at $e$.}

\begin{exmp}
In Figure~\ref{Mergings}(a), paths $\beta_1$ and $\beta_2$ share some vertex, however not edges/subpaths, so $\beta_1$ and $\beta_2$ do not merge. In Figure~\ref{Mergings}(b), paths $\beta_1$ and $\beta_2$ do share edge $S \to T$, where $S$ is a source, however condition $2$ is not satisfied, therefore $\beta_1$ and $\beta_2$ do not merge, although they internally intersect at $S \to T$. In Figure~\ref{Mergings}(c), $\beta_1$ and $\beta_2$ merge at edge $A \to B$, at subpath $A \to B \to C$; $\beta_2$ and $\beta_3$ merge at edge $A \to B$, at subpath $A \to B \to C \to D$; $\beta_1$, $\beta_2$ and $\beta_3$ merge at edge $A \to B$, at subpath $A \to B \to C$; $\beta_4$ merges with $\beta_3$ at edge $B \to C$, at subpath $B \to C \to D$; there are two mergings in Figure~\ref{Mergings}(c), at edge $A \to B$, and at edge $B \to C$, respectively.
\begin{figure}
\psfrag{b1}{$\beta_1$} \psfrag{b2}{$\beta_2$} \psfrag{b3}{$\beta_3$} \psfrag{b4}{$\beta_4$}
\centerline{\includegraphics[width=4in]{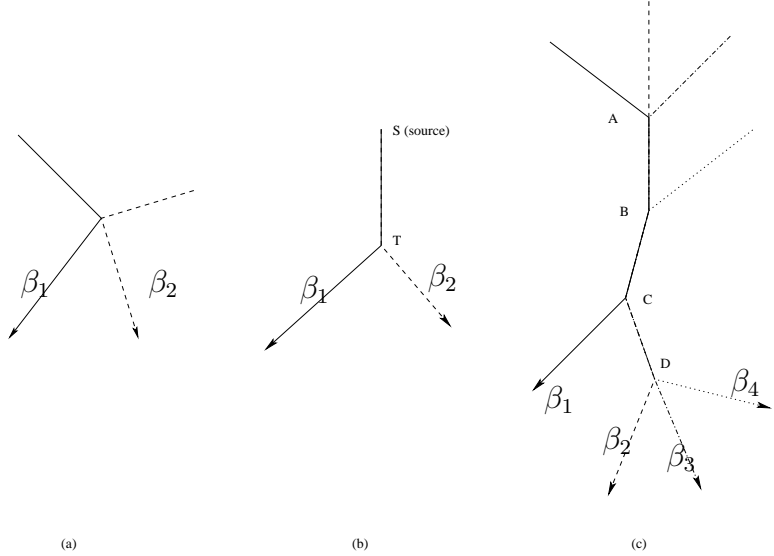}}
\caption{examples of mergings and non-mergings}
\label{Mergings}
\end{figure}
\end{exmp}

In this paper, we will consider an acyclic directed graph $G(E, V)$ with $n$ sources and $n$ sinks. Unless specified otherwise, we will use $S_1, S_2, \cdots, S_n$ to denote the sources and $R_1, R_2, \cdots, R_n$ to denote the sinks; $c_i$ will be used to denote the min-cut between $S_i$ and $R_i$, and $\alpha_i=\{\alpha_{i, 1}, \alpha_{i, 2}, \cdots, \alpha_{i, c_i}\}$ will be used to denote a set of Menger's paths from $S_i$ and $R_i$. We will study how $\alpha_i$'s merge with each other; more specifically, we show that appropriately chosen Menger's paths will only merge with each other finitely many times. In particular, we deal with the case when all sources and sinks are distinct in Section~\ref{M}, and the case when the sources are identical and the sinks are distinct in Section~\ref{MStar}. For both of cases, we will study how the minimum merging number depends on the min-cuts.

We remark that when $n=1$, Ford-Fulkerson algorithm~\cite{fo56} can find the min-cut and a set of Menger's path between $S_1$ and $R_1$ in polynomial time. The LDP (Link Disjoint Problem) asks if there are two edge-disjoint paths from $S_1$, $S_2$ to $R_1$, $R_2$, respectively. The fact that the LDP problem is NP-complete~\cite{fo80} suggests the intricacy of the problem when $n \geq 2$.

\medskip

{\bf Notation and Convention.} For a path $\gamma$ in an acyclic direct graph $G$, let $a(\gamma), b(\gamma)$ denote the starting point and the ending point of $\gamma$, respectively; let $\gamma[s, t]$ denote the subpath of $\gamma$ with the starting point $s$ and the ending point $t$. For two distinct paths $\gamma, \pi$ in $G$, we say $\gamma$ is {\em smaller} than $\pi$ if there is a directed path from $b(\gamma)$ to $a(\pi)$; if $\gamma, \pi$ and the connecting path from $b(\gamma)$ to $a(\pi)$ are subpaths of path $\beta$, we say $\gamma$ is {\em smaller} than $\pi$ on $\beta$. Note that this definition also applies to the case when paths degenerate to vertices/edges; in other words, in the definition, $\gamma, \pi$ or the connecting path from $b(\gamma)$ to $a(\pi)$ can be vertices/edges in $G$, which can be viewed as degenerated paths. If $b(\gamma)=a(\pi)$, we use $\gamma \circ \pi$ to denote the path obtained by concatenating $\gamma$ and $\pi$ subsequently. For a set of vertices $v_1, v_2, \cdots, v_j$ in $G$, define $G|v_1, \cdots, v_j)$ to be subgraph of $G$ consisting of the set of vertices (denoted by $V_0$), each of which is smaller than some $v_j$, and the set of all the edges, each of which is incident with some vertex in $V_0$.

\section{Minimum Mergings $\mathcal{M}$} \label{M}

In this section, we consider any acyclic directed graph $G$ with $n$ distinct sources and $n$ distinct sinks. Let $M(G)$ denote the minimum number of mergings over all possible Menger's path sets $\alpha_i$'s, $i=1, 2, \cdots, n$, and let $\mathcal{M}(c_1, c_2, \cdots, c_n)$ denote the supremum of $M(G)$ over all possible choices of such $G$.

In the following, we shall prove that
\begin{thm} \label{main}
For any $c_1, c_2, \cdots, c_n$,
$$
\mathcal{M}(c_1, c_2, \cdots, c_n) < \infty,
$$
and furthermore, we have
$$
\mathcal{M}(c_1, c_2, \cdots, c_n) \leq \sum_{i < j} \mathcal{M}(c_i, c_j).
$$
\end{thm}

Now consider
$$
\alpha_i=\{\alpha_{i, 1}, \alpha_{i, 2}, \cdots, \alpha_{i, c_i}\},
$$
a set of Menger's paths from $S_i$ to $R_i$, and
$$
\alpha_j=\{\alpha_{i, 1}, \alpha_{i, 2}, \cdots, \alpha_{i, c_j}\},
$$
a set of Menger's paths from $S_j$ to $R_j$. For two merged subpaths $u, v$ by $\alpha_i$ and $\alpha_j$ (more rigorously, by some paths from $\alpha_i$ and $\alpha_j$), we say $v$ is {\it semi-reachable through $\alpha_i$} by $u$ if there is a sequence of merged subpaths $\gamma_0, \gamma_1, \cdots, \gamma_n$ by $\alpha_i$ and $\alpha_j$ such that
\begin{enumerate}
\item $\gamma_0=u$, $\gamma_n=v$;
\item For each feasible $k$, $\gamma_{2k+1}$ is smaller than $\gamma_{2k}$ on some $\alpha_{j, t_k}$, and $\alpha_{j, t_k}[b(\gamma_{2k+1}), a(\gamma_{2k})]$ doesn't merge with any paths from $\alpha_i$;
\item For each feasible $k$, $\gamma_{2k+1}$ is smaller than $\gamma_{2k+2}$ on some $\alpha_{i, h_k}$.
\end{enumerate}
We say $v$ is {\it regularly-semi-reachable through $\alpha_i$} by $u$ if besides the three conditions above, we further require that all $h_k$'s in condition $3$ are distinct from each other. If $n$ is an even number, we say $v$ is semi-reachable through $\alpha_i$ by $u$ {\em from above}; if $n$ is an odd number, we say $v$ is semi-reachable through $\alpha_i$ by $u$ {\em from below} (``above'' and ``below'' naturally come up when $G$ is ``drawn'' in an $(x,y)$-plane such that smaller paths are always higher than larger paths, as exemplified in Figure~\ref{reachable}). It immediately follows that for three merged subpaths $u, v, w$ by $\alpha_i, \alpha_j$, if $v$ is semi-reachable through $\alpha_i$ from above by $u$, $w$ is semi-reachable through $\alpha_i$ from above by $v$, then $w$ is also semi-reachable through $\alpha_i$ from above by $u$.

The following two propositions are more or less obvious.
\begin{pr}  \label{quasi}
Consider Menger's path sets $\alpha_i, \alpha_j$ and merged subpaths by $\alpha_i, \alpha_j$. For a merged subpath $v$ semi-reachable through $\alpha_i$ by a merged subpath $u$ via a sequence of merged subpaths $\gamma_0, \gamma_1, \cdots, \gamma_n$, if none of $\gamma_i$'s is semi-reachable through $\alpha_i$ by itself from above, then $v$ is regularly-semi-reachable through $\alpha_i$ by $u$.
\end{pr}

To see this, consider any $k < l$ (if any) such that $h_k=h_l$ and $h_k, h_{k+1}, h_{k+2}, \cdots, h_{l-1}$ are all distinct from each other. Since none of $\gamma_i$'s is semi-reachable through $\alpha_i$ by itself from above, one checks that $v$ is semi-reachable through $\alpha$ via a shorter intermediate sequence of merged subpaths
$$
\gamma_0, \cdots, \gamma_{2k+1}, \gamma_{2l+2}, \cdots, \gamma_n.
$$
Continue to find such shorter intermediate sequences iteratively in the similar fashion until all $h_k$'s (corresponding to the new intermediate sequence) are all distinct from each other.

\begin{pr} \label{crossing}
Consider Menger's path sets $\alpha_i, \alpha_j$ and merged subpaths by $\alpha_i, \alpha_j$. If a merged subpath $u$ is semi-reachable through $\alpha_i$ by itself from above via a sequence of merged subpaths $\gamma_0, \gamma_1, \cdots, \gamma_{2m}=\gamma_0$, then one can find a new set, still denoted by $\alpha_i$, of $m$ pairwise edge-disjoint paths from $S_i$ to $R_i$ such that the number of mergings between $\alpha_j$ and the new $\alpha_i$ strictly decreases.
\end{pr}

To see this, consider the following {\it reroutings} of all $\alpha_i$-paths, each of which reaches some $b(\gamma_{2k+1})$: when any of such $\alpha_i$-path, say $\alpha_{i, l}$, reaches $b(\gamma_{2k+1})$, instead of continuing on its original ``trajectory'', it is rerouted to continue on $\alpha_{j, t_k}[b(\gamma_{2k+1}), b(\gamma_{2k})]$, and then from $b(\gamma_{2k})$ it continues on the $\alpha_i$-path (typically different from $\alpha_{i, l}$) incident with $b(\gamma_{2k})$. For example, for the case when $u$ is regularly-semi-reachable through $\alpha_i$ by itself from above, one can reroute $\alpha_i$ to obtain a set of $m$ pairwise edge-disjoint paths from $S_i$ to $R_i$, by replacing $\alpha_{i, h_k}[{b(\gamma_{2k+1})}, R_i]$ by $\alpha_{j, t_k}[{b(\gamma_{2k+1})}, a(\gamma_{2k})] \circ \alpha_{i, h_{k-1}}[a(\gamma_{2k}), R_i]$ for all feasible $k$ (here $h_{0} \stackrel{\triangle}{=}h_{m}$).

Note that the above reroutings ``desert'' certain subpaths in the original $\alpha_i$ and ``borrow'' other subpaths from $\alpha_j$ to obtain a new Menger's path set $\alpha_i$ from $S_i$ to $R_i$. One checks that after such reroutings, the number of mergings between $\alpha_i$ and $\alpha_j$ strictly decreases (however the number of mergings in $G$ may remain the same). We say that $\alpha_i$ is {\em reroutable} using $\alpha_j$ if there exist certain reroutings of some $\alpha_i$-paths using some $\alpha_j$-paths; we say $G$ is {\em reroutable} with respect to $\alpha_1, \alpha_2, \cdots, \alpha_n$ if there exist $i \neq j$ such that $\alpha_i$ is is reroutable using $\alpha_j$.

The following proposition deals with the opposite direction of Proposition~\ref{crossing} for the case when $G$ has $2$ distinct sources and $2$ distinct sinks.
\begin{pr}
Consider the case when there are $2$ distinct sources and $2$ distinct sinks in $G$. For any rerouting of $\alpha_1$ using $\alpha_2$-subpaths, there is a merged subpath semi-reachable through $\alpha_1$ by itself from above.
\end{pr}

\begin{proof}

Assume that subpaths $\gamma_1, \gamma_2, \cdots, \gamma_l$ are the ``deserted'' subpaths for a given rerouting of $\alpha_1$, and these subpaths ``spread'' out to $\alpha_{1, 1}, \alpha_{1, 2}, \cdots, \alpha_{1, k}$, $k \leq l$. Without loss of generality, further assume that $\gamma_1, \gamma_2, \cdots, \gamma_{k}$ are the smallest such deserted subpaths on $\alpha_{1, 1}, \alpha_{1, 2}, \cdots, \alpha_{1, k}$, respectively. Then there are $\alpha_2$-subpaths $\eps_1, \eps_2, \cdots, \eps_{k}$ such that all $\eps_i$'s do not merge with any $\alpha_1$-paths, and for each $i$ with $1 \leq i \leq k$ and correspondingly certain $k_i$ with $1 \leq k_i \leq l$, $a(\eps_i)=a(\gamma_i)$, $b(\eps_i)=b(\gamma_{k_i})$. Surely one can find a subset $\{\hat{k}_1, \hat{k}_2, \cdots, \hat{k}_s\}$ of $\{1, 2, \cdots, k\}$ such that $b(\eps_1) \in \alpha_{1, \hat{k}_1}$, $b(\eps_{\hat{k}_1}) \in \alpha_{1, \hat{k}_2}$, $\cdots$, $b(\eps_{\hat{k}_s}) \in \alpha_{1, 1}$, which implies that there is a merged subpath (for instance, the one merged by $\gamma_1$ and $\eps_{\hat{k}_s}$) semi-reachable through $\alpha_1$ by itself from above.

\end{proof}

\begin{rem}  \label{imaginary}
Consider any set of edge-disjoint paths $\beta=\{\beta_1, \beta_2, \cdots, \beta_m\}$ in $G$. If we add ``imaginary'' source $S$ together with $m$ disjoint edges from $S$ to all $a(\beta_i)$'s, and add ``imaginary'' sink $R$ together with $m$ disjoint edges from all $b(\beta_i)$'s to $R$, we obtain a set of Menger's paths from $S$ to $R$ in the graph extended from $G$. In this section, we don't differentiate between a set of Menger's paths and a set of edge-disjoint paths for simplicity, since we can always assume the existence of such imaginary sources and sinks when they are needed.
\end{rem}

\begin{exmp}
\psfrag{a1}{$\alpha_{i, 1}$} \psfrag{a2}{$\alpha_{i, 2}$} \psfrag{a3}{$\alpha_{i, 3}$}
\psfrag{b1}{$\alpha_{j, 1}$} \psfrag{b2}{$\alpha_{j, 2}$} \psfrag{b3}{$\alpha_{j, 3}$} \psfrag{b4}{$\alpha_{j, 4}$}
\psfrag{g0}{$\gamma_0$} \psfrag{g1}{$\gamma_1$} \psfrag{g2}{$\gamma_2$} \psfrag{g3}{$\gamma_3$} \psfrag{g4}{$\gamma_4$} \psfrag{g5}{$\gamma_5$} \psfrag{g}{$\gamma$}
\begin{figure}
\centerline{\includegraphics[width=2.5in]{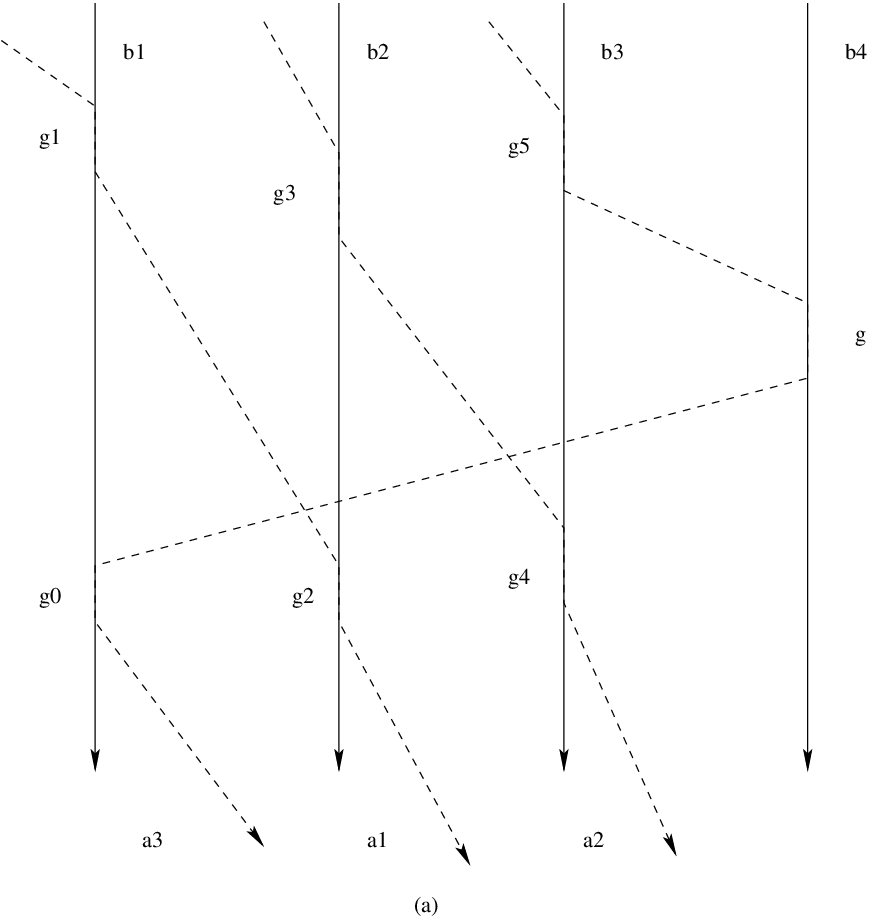} \hspace{0.5cm} \includegraphics[width=2.5in]{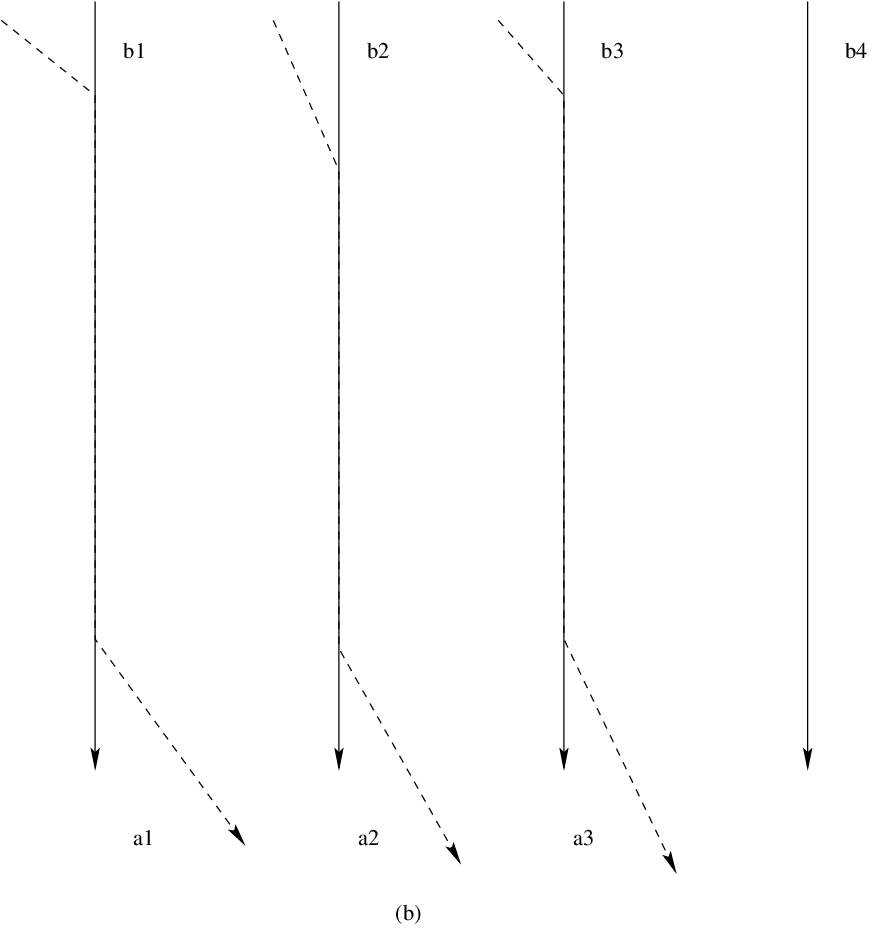}}
\caption{an example}
\label{reachable}
\end{figure}

In Figure~\ref{reachable}(a), $\gamma$ and $\gamma_i$ ($i=0, 1, \cdots, 5$) are merged subpaths from $\alpha_i=\{\alpha_{i, 1}, \alpha_{i, 2}, \alpha_{i, 3}\}$ and $\alpha_j=\{\alpha_{j, 1}, \alpha_{j, 2}, \alpha_{j, 3}, \alpha_{j, 4}\}$. By definitions, we have
\begin{enumerate}
\item $\gamma_1, \gamma_3, \gamma_5$ are semi-reachable through $\alpha_i$ from below by $\gamma_0$.
\item $\gamma_3, \gamma_5$ are semi-reachable through $\alpha_i$ from below by $\gamma_2$.
\item $\gamma_2, \gamma_4$ are semi-reachable through $\alpha_i$ from above by $\gamma_0$.
\item $\gamma$ is semi-reachable through $\alpha_i$ from above by $\gamma_0, \gamma_2, \gamma_4$.
\item $\gamma_0$ is semi-reachable through $\alpha_i$ from above by itself (via the sequence of merged subpaths $\gamma_0, \gamma_1, \gamma_2, \gamma_3, \gamma_4, \gamma_5, \gamma_0$) , so are $\gamma_2, \gamma_4$, thus, as shown in Figure~\ref{reachable}(b), $\alpha_i$ is routable using $\alpha_j$ by Proposition~\ref{crossing}.
\end{enumerate}

\end{exmp}

Before the proof of Theorem~\ref{main}, we shall first prove the following lemma.

\begin{lem}  \label{twotwo}
For any $c_1, c_2$,
$$
\mathcal{M}(c_1, c_2) \leq c_1c_2(c_1+c_2)/2.
$$
\end{lem}

\begin{proof}
Consider any acyclic directed graph $G(E, V)$ with $2$ distinct sources $S_1, S_2$ and $2$ distinct sinks $R_1, R_2$, where the min-cut between $S_i$ and $R_i$ is $c_i$ for $i=1, 2$. Let $\alpha_1=\{\alpha_{1, 1}, \cdots, \alpha_{1, c_1}\}$ be any set of Menger's paths from $S_1$ to $R_1$, and $\alpha_2=\{\alpha_{2, 1}, \cdots, \alpha_{2, c_2}\}$ be any set of Menger's paths from $S_2$ to $R_2$. Let $V_{\mathcal{M}}$ be the set of the terminal vertices (starting and ending vertices) of all the merged subpaths by $\alpha_1$ and $\alpha_2$. It suffices to prove that for any $c_1, c_2$, if $|V_{\mathcal{M}}| \geq c_1c_2(c_1+c_2)+1$, one can always reroute $\alpha_1$ using $\alpha_2$, or reroute $\alpha_2$ using $\alpha_1$ to obtain new Menger's path sets $\alpha_1, \alpha_2$ such that the number of mergings between the new $\alpha_1, \alpha_2$ is strictly less than that between the original $\alpha_1, \alpha_2$.

Now we perform certain operations on $G$ to obtain another graph $\hat{G}$. First we delete all the edges which do not belong to any $\alpha_1$-path or $\alpha_2$-path; then whenever two paths $\beta_1, \beta_2$ from $\alpha_1 \cup \alpha_2$ ($\beta_1, \beta_2$ could be both $\alpha_1$-paths or $\alpha_2$-paths) intersect on a vertex $v$, however do not share any edge incident with $v$ (for an example, see Figure~\ref{Mergings}(a)), we ``detach'' $\beta_1, \beta_2$ at $v$ (in other words, ``split'' $v$ into two copies $v^{(1)}, v^{(2)}$ and let $\beta_1$ pass $v^{(1)}$ and let $\beta_2$ pass $v^{(2)}$); next we delete all the merged subpaths by $\alpha_1$ and $\alpha_2$; finally we reverse the direction of the edges which only belong to some $\alpha_2$-path. Note that the above operations does not add more vertices to $G$; and for any path in $\hat{G}$, each edge either belongs to a $\alpha_1$-path or a reversed $\alpha_2$-path.

Suppose that there is a cycle in $\hat{G}$ taking the following form:
$$
\gamma_1 \circ \gamma_2 \circ \cdots \circ \gamma_{2n},
$$
where $b(\gamma_{2n})=a(\gamma_1)$, $\gamma_i$ is a reversed $\alpha_2$-subpath for any odd $i$ and a $\alpha_1$-subpath for any even $i$. For any vertex $w$ in $V_{\mathcal{M}}$, let $\eps_w$ denote the merged subpath in $G$ corresponding to $w$; then one checks that in $G$, $\eps_{a(\gamma_1)}$ is semi-reachable through $\alpha_1$ by itself from above via the sequence
$$
\eps_{a(\gamma_1)}, \eps_{a(\gamma_2)}, \cdots, \eps_{a(\gamma_{2n})}, \eps_{b(\gamma_{2n})},
$$
which implies certain reroutings can be done to reduce the number of mergings.

Next we assume that $\hat{G}$ is acyclic. Note that in $\hat{G}$, $S_1, R_2$ have out-degree $c_1, c_2$, respectively, $S_2, R_1$ has in-degree $c_1, c_2$, respectively, and any vertex in $V_{\mathcal{M}}$ has in-degree $1$ and out-degree $1$. It then immediately follows that $\hat{G}$ consists of $c_1+c_2$ pairwise vertex-disjoint paths, each of which, say $\gamma$, takes the following {\em regular} form:
$$
\gamma=\gamma_1 \circ \gamma_2 \circ \cdots \circ \gamma_n,
$$
where $a(\gamma_1)=S_1 \mbox{ or } R_2$, $b(\gamma_n)=S_2 \mbox{ or } R_1$, the terminal points of $\gamma_2, \gamma_3, \cdots, \gamma_{n-1}$ are in $V_{\mathcal{M}}$, and each of $\gamma_1,\gamma_2, \cdots, \gamma_n$ is, alternately, either a $\alpha_1$-subpath or a reversed $\alpha_2$-subpath. Since $|V_{\mathcal{M}}| \geq c_1c_2(c_1+c_2)+1$, out of the $c_1+c_2$ pairwise edge-disjoint paths, there must be at least one path, say $\gamma$, taking the regular form $\gamma=\gamma_1 \circ \gamma_2 \circ \cdots \circ \gamma_n$, such that $|V_\mathcal{M} \cap \gamma| \geq c_1c_2+1$. It then follows that there are two vertices $u, v \in V_{\mathcal{M}}$ on $\gamma$, where $u$ corresponds to the merged subpath by $\alpha_{1, i_1}$ and $\alpha_{2, j_1}$, and $v$ corresponds to the merged subpath by $\alpha_{1, i_2}$ and $\alpha_{2, j_2}$, such that $(i_1, j_1)=(i_2, j_2)$. Note that if $u$ is larger (smaller) than $v$ on $\alpha_{1, i_1}$, then $u$ will be also larger (smaller) than $v$ on $\alpha_{2, j_1}$, otherwise we would have a cycle $\alpha_{1, i_1}[u, v] \circ \alpha_{1, j_1}[v, u]$ in $G$, which contradicts the assumption that $G$ is acyclic. Now assume that $\gamma[u, v]=\gamma_s \circ \gamma_{s+1} \circ \cdots \circ \gamma_t$. First consider the following conditions (ignoring the parathetic words for the moment):
\begin{itemize}
\item $u$ is smaller (larger) than $v$ on $\alpha_{1, i_1}$;
\item $\gamma_i$ is a $\alpha_1$-subpath (reversed $\alpha_2$-subpath) for $i=s+1$;
\item $u$ is the starting (ending) vertex of the corresponding merged subpath in $G$, $v$ is the starting (ending) vertex of the corresponding merged subpath in $G$.
\end{itemize}
Then one checks that $\eps_v$ is semi-reachable by itself from above through $\alpha_2$ via the sequence $\eps_v, \eps_{b(\gamma_{t-1})}, \cdots, \eps_{b(\gamma_{s})}, \eps_u, \eps_v$, implying a rerouting of $\alpha_2$ using $\alpha_1$ to reduce the number of mergings can be done. Similar arguments can be applied to other cases when any parathetic words replace the words before them.

So in any case, if $|V_{\mathcal{M}}| \geq c_1c_2(c_1+c_2)+1$, certain reroutings can be done to strictly reduce the number of mergings. Together with the fact that each merged subpath has two terminal points, we then prove that $\mathcal{M}(c_1, c_2) \leq c_1c_2(c_1+c_2)/2$, establishing the lemma.

\end{proof}

We are now ready for the proof of Theorem~\ref{main}.

\begin{proof}

With Lemma~\ref{twotwo} being established, to prove Theorem~\ref{main}, it suffices to prove that
\begin{equation} \label{recursive}
\mathcal{M}(c_1, c_2, \cdots, c_n) \leq \mathcal{M}(c_1, c_2, \cdots, c_{n-1}) + \sum_{i < n} \mathcal{M}(c_i, c_n),
\end{equation}
for $n=3, 4, \cdots,$ inductively.

Now suppose that for $n \leq k$, $\mathcal{M}(c_1, c_2, \cdots, c_n)$ is finite and satisfies (\ref{recursive}) and consider the case $n=k+1$. For $i=1, 2, \cdots, k+1$, choose a set of Menger's paths $\alpha_i=\{\alpha_{i, 1}, \alpha_{i, 2}, \cdots, \alpha_{i, c_i}\}$ between $S_i$ and $R_i$, and assume $\alpha_1, \alpha_2, \cdots, \alpha_k$ are chosen such that the number of mergings among themselves is no more than $\mathcal{M}(c_1, c_2, \cdots, c_k)$. By a ``new'' merging, we mean a merging which is among $\alpha_1, \alpha_2, \cdots, \alpha_{k+1}$, however is not among $\alpha_1, \alpha_2, \cdots, \alpha_{k}$. We shall prove that if the number of new mergings between $\alpha_{k+1}$ and $\alpha_1, \alpha_2, \cdots, \alpha_{k}$ is larger than or equal to
$$
\mathcal{M}(c_1, c_{k+1})+\mathcal{M}(c_2, c_{k+1})+\cdots+\mathcal{M}(c_k, c_{k+1})+1,
$$
certain reroutings can be done to strictly reduce the number of mergings.

By contradiction, assume the opposite of the claim above and label all the newly merged subpaths as $\gamma_1, \gamma_2, \cdots, \gamma_l$. By the Pigeonhole principle, there exists some $\alpha_i$ such that $\alpha_i$ and $\alpha_{k+1}$ will have more than $\mathcal{M}(c_i, c_{k+1})$ new mergings, thus reroutings of $\alpha_i$ or $\alpha_{k+1}$ can be done. If such a rerouting is in fact a rerouting of $\alpha_{k+1}$ using $\alpha_i$, then the number of mergings between $\alpha_{k+1}$ and $\alpha_1, \alpha_2, \cdots, \alpha_{k}$ will be strictly decreased after the rerouting. So in the following we assume that the rerouting between every $\alpha_i$ and $\alpha_{k+1}$, if exists, is a rerouting of $\alpha_i$ using $\alpha_{k+1}$. Then after the rerouting of $\alpha_i$, the new $\alpha_i$ will ``miss'' at least
$$
\mathcal{M}(c_1, c_{k+1})+\cdots+\mathcal{M}(c_{i-1}, c_{k+1})+\mathcal{M}(c_{i+1}, c_{k+1})+\cdots+\mathcal{M}(c_k, c_{k+1})+1
$$
of all the newly merged subpaths, which implies the new $\alpha_j$'s, $j \leq k$, will all ``miss'' at least one of newly merged subpaths (in other words, there is $\gamma_{l_0}$ such that none of $\alpha_j$'s, $j \leq k$, merge with $\alpha_{k+1}$ at $\gamma_{l_0}$). So the number of mergings between $\alpha_1, \alpha_2, \cdots, \alpha_{k}$ and $\alpha_{k+1}$ strictly decreases after the possible reroutings of all $\alpha_i$'s. With this contradiction, we establish the theorem.

\end{proof}

\begin{rem}
For an acyclic directed graph $G(V, E)$, the vertex-connectivity version of Menger's theorem~\cite{Menger1927} states:

{\em For any $u, v \in V$, with no edge from $u$ to $v$, the maximum number of pairwise vertex-disjoint directed paths from $u$ to $v$ in $G$ equals the minimum vertex cut between $u$ and $v$, namely the minimum number of vertices in $E \backslash \{u, v\} $ whose deletion destroys all directed paths from $u$ to $v$.}

In this remark, we redefine Menger's paths and merging: we call any set consisting of the maximum number of pairwise vertex-disjoint directed paths from $u$ to $v$ a set of {\em Menger's paths} from $u$ and $v$; and for $m$ paths $\beta_1, \beta_2, \cdots, \beta_m$ in $G(V, E)$, we say these paths {\it merge} at $e \in V$ (here $E$ in the original definition is replaced by $V$) if
\begin{enumerate}
\item $e \in \cap_{i=1}^m \beta_i$;
\item there are at least two distinct $f, g \in E$ such that $f, g$ are immediately ahead of $e$ on some $\beta_i, \beta_j$, respectively.
\end{enumerate}
And naturally we can also redefine $\mathcal{M}$ with the above redefined Menger's paths and merging. Then using a parallel argument, one can show that Theorem~\ref{main} still hold true for redefined $\mathcal{M}$.

\end{rem}

The following proposition shows that $\mathcal{M}$ is symmetric on its parameters.

\begin{pr} \label{symmetric}
For any $c_1, c_2, \cdots, c_n$, we have
$$
\mathcal{M}(c_1, c_2, \cdots, c_n)=\mathcal{M}(c_{\delta(1)}, c_{\delta(2)}, \cdots, c_{\delta(n)}),
$$
where $\delta$ is any permutation on the set $\{1, 2, \cdots, n\}$.
\end{pr}

The following proposition shows that $\mathcal{M}$ is an ``increasing'' function.

\begin{pr}  \label{increasing}
For any $m \geq n$, $c_1 \leq c_2 \leq \cdots \leq c_n$ and $d_1 \leq d_2 \cdots \leq d_m$, if
$c_i \leq d_{m-n+i}$ for $i=1, 2, \cdots, n$, then
$$
\mathcal{M}(c_1, c_2, \cdots, c_n) \leq \mathcal{M}(d_1, d_2, \cdots, d_m).
$$
\end{pr}

Together with Proposition~\ref{symmetric}, the following proposition shows that when $\mathcal{M}$ has two parameters, $\mathcal{M}$ is ``sup-linear'' in all its parameters.

\begin{pr}  \label{Isolated}
For any $c_{1, 0}, c_{1, 1}, c_2$, we have
$$
\mathcal{M}(c_{1, 0}+c_{1, 1}, c_2) \geq \mathcal{M}(c_{1, 0}, c_2)+\mathcal{M}(c_{1, 1}, c_2).
$$
\end{pr}

\begin{proof}

For any $c_{1, 0}, c_{1, 1}$ and $c_2$, consider the following directed graph $G$ with $2$ sources $S_1, S_2$ and $2$ sinks $R_1, R_2$ such that
\begin{enumerate}
\item there is a set $\alpha_1$ of $c_{1, 0}+c_{1, 1}$ edge-disjoint paths from $S_1$ to $R_1$, here $\alpha_1=\alpha_1^{(0)} \cup \alpha_1^{(1)}$, where $\alpha_1^{(0)}$ and $\alpha_1^{(1)}$ are mutually exclusive, consisting of $c_{1, 0}$, $c_{1, 1}$ edge-disjoint paths, respectively, and there is a set $\alpha_2$ of $c_2$ edge-disjoint paths from $S_2$ to $R_2$;
\item mergings by $\alpha_1^{(0)}, \alpha_2$ and mergings by $\alpha_1^{(1)}, \alpha_2$ are ``sequentially isolated'' on $\alpha_2$ in the sense that on each $\alpha_2$-path, the smallest merged $\alpha_1^{(1)}$-subpath is larger than the largest merged $\alpha_1^{(0)}$-subpath;
\item the number of mergings in the subgraph consisting of $\alpha_1^{(0)}$ and $\alpha_2$ achieves $\mathcal{M}(c_{1, 0}, c_2)$, and the number of mergings in the subgraph consisting of $\alpha_1^{(1)}$ and $\alpha_2$ achieves $\mathcal{M}(c_{1, 1}, c_2)$.
\end{enumerate}

One checks that for such graph $G$, the min-cut between $S_1$ and $R_1$ is $c_{1, 0}+c_{1, 1}$, and the min-cut between $S_2$ and $R_2$ is $c_2$, and
$$
M(G)=\mathcal{M}(c_{1, 0}, c_2)+\mathcal{M}(c_{1, 1}, c_2),
$$
which implies that
$$
\mathcal{M}(c_{1, 0}+c_{1, 1}, c_2) \geq \mathcal{M}(c_{1, 0}, c_2)+\mathcal{M}(c_{1, 1}, c_2).
$$

\end{proof}

\begin{pr} \label{BraveHeart}
For any $c_1, c_2, \cdots, c_{n}$ and any fixed $k$ with $1 \leq k \leq n$, we have
$$
\mathcal{M}(c_1, c_2, \cdots, c_n) \geq \sum_{i \leq k, j \geq k+1} \mathcal{M}(c_i, c_j).
$$
\end{pr}

\begin{proof}

For any $c_1, c_2, \cdots, c_n$, consider the following directed graph $G$ with $n$ sources $S_1, S_2, \cdots, S_n$ and $n$ sinks $R_1, R_2, \cdots, R_n$ such that for any fixed $k$ with $1 \leq k \leq n$,
\begin{enumerate}
\item there is a set $\alpha_i$ of $c_i$ edge-disjoint paths from $S_i$ to $R_i$ for each $i$;
\item all $\alpha_i$'s, $i \leq k$, do not merge with each other, and all $\alpha_j$'s, $j \geq k+1$, do not merge with each other;
\item for any $i$ with $i \leq k$, mergings by $\alpha_i$ and all $\alpha_j$'s, $j \geq k+1$, are ``sequentially isolated'' on $\alpha_i$ in the sense that on each $\alpha_i$-path, for any $j_1 < j_2$ with $j_1, j_2 \geq k+1$, the smallest merged $\alpha_{j_2}$-subpath is larger than the largest merged $\alpha_{j_1}$-subpath. Similarly for any $j$ with $j \geq k+1$, mergings by $\alpha_j$ and all $\alpha_i$'s, $i \leq k$, are sequentially isolated on $\alpha_j$;
\item the number of mergings in the subgraph consisting of any $\alpha_i$ with $i \leq k$ and any $\alpha_j$ with $j \geq k+1$ achieves $\mathcal{M}(c_i, c_j)$.
\end{enumerate}

One checks that for such graph $G$, the min-cut between $S_i$ and $R_i$ is $c_i$, and
$$
M(G)=\sum_{i \leq k, j \geq k+1} \mathcal{M}(c_i, c_j),
$$
which implies that
$$
\mathcal{M}(c_1, c_2, \cdots, c_n) \geq \sum_{i \leq k, j \geq k+1} \mathcal{M}(c_i, c_j).
$$

\end{proof}

The following proposition gives an upper bound on $\mathcal{M}(m, n)$ using $\mathcal{M}(m_1, n_1)$'s, where $m_1 \leq m$, $n_1 \leq n$.

\begin{pr} \label{ILikeIt}
For any $m \leq n$, we have
$$
\mathcal{M}(m, n) \leq U(m, n)+V(m, n)+m-2,
$$
where
$$
U(m, n)=\sum_{j=1}^{m-1} \left(\mathcal{M}(j, m-1)+1+\mathcal{M}(m-j, n)\right) + \mathcal{M}(m, m-1)+1,
$$
and
$$
V(m, n)=\mathcal{M}(m, n-1)+\sum_{j=1}^{m-1} \left(\mathcal{M}(j, n)+1+\mathcal{M}(m-j, n)\right)-\mathcal{M}(1, n).
$$
\end{pr}

\begin{proof}

Consider any acyclic directed graph $G(E, V)$ with $2$ distinct sources $S_1, S_2$ and $2$ distinct sinks $R_1, R_2$. Assume the min-cut between $S_1$ and $R_1$ is $m$, and the min-cut between $S_2$ and $R_2$ is $n$. Let $\phi=\{\phi_1, \cdots, \phi_m\}$ be any set of Menger's paths from $S_1$ to $R_1$, and $\psi=\{\psi_1, \cdots, \psi_n\}$ be any set of Menger's paths from $S_2$ to $R_2$. Let $|G|_{\mathcal{M}}$ denote the number of mergings in $G$ (in this proof, we only consider mergings by $\phi$ and $\psi$). It suffices to prove that for any $m, n$, if
\begin{equation}  \label{bigN}
|G|_{\mathcal{M}} \geq U(m, n)+V(m, n)+m-1,
\end{equation}
then $G$ is reroutable with respect to $\phi$ and $\psi$, namely, one can always reroute $\phi$ using $\psi$, or reroute $\psi$ using $\phi$ to obtain new Menger's path sets $\phi$ and $\psi$.

By contradiction, assume that even if (\ref{bigN}) is satisfied, there are no reroutings to reduce the number of mergings. In the following, we say a merged subpath $\gamma_1$ is {\em immediately ahead of} another merged subpath $\gamma_2$ (or $\gamma_2$ is {\em immediately behind} $\gamma_1$) on certain path $\beta$ if $\gamma_1$ is smaller than $\gamma_2$ on $\beta$ and there is no other merged subpath in between $\gamma_1$ and $\gamma_2$ on $\beta$.

Consider the following iterative procedure, where, for notational simplicity, we treat a graph as a union of its vertex set and edge set. Let $T_0$ be the initial graph only consisting of $S_1, S_2$. We will subsequently construct a sequence of graphs $T_1, T_2, \cdots$ such that $T_i \subset T_{i+1}$ for feasible $i$. Suppose we have obtained $T_i$. Now pick a merged subpath $\gamma_{i+1}$ outside $T_i$ such that each merged subpath within $G|b(\gamma_{i+1})) \backslash T_i$ (here $\backslash$ is the symbol for ``relative complement'' in set theory) is immediately behind some merged subpath in $T_i$ on some $\phi$-path (if $i=0$, treat $S_1, S_2$ as degenerated merged subpaths). Define $T_{i+1}=G|b(\gamma_{i+1})) \cup T_i$. One checks that when $|T_i|_{\mathcal{M}} < |G|_{\mathcal{M}}$ such $\gamma_{i+1}$ always exists; and $|T_{i+1}|_{\mathcal{M}}-|T_i|_{\mathcal{M}} \leq m$, where again $|T_i|_{\mathcal{M}}$ denotes the number of mergings in $T_i$. So there exists $l$ such that $|T_l|_{\mathcal{M}} < V(m, n)$ and $|T_{l+1}|_{\mathcal{M}} \geq V(m, n)$. Now for each $i$ let $\eps_i$ be the largest merged subpath in $T_{l+1}$ on $\psi_i$. Note that such $\eps_i$ always exists, since otherwise $T_{l+1} \backslash \beta_i$, consisting of subpaths from $m$ $\phi$-paths and $n-1$ $\psi$-paths, will have more than $\mathcal{M}(m, n-1)$ mergings, which will lead to certain merging reducing reroutings. Let
$$
\mathbb{S}\stackrel{\triangle}{=}T_{l+1}=G|b(\eps_1), \cdots, b(\eps_n)).
$$
Then $|\mathbb{S}|_{\mathcal{M}} \geq V(m, n)$ and $|\mathbb{R}=G \backslash \mathbb{S}|_{\mathcal{M}} \geq U(m, n)$.

Now arbitrarily pick $j_1$ and assume that within $\mathbb{R}$, $\phi_{j_1}$ merges with $\psi$ at the merged subpaths $\eta_1^{(1)}, \eta_2^{(1)}, \cdots, \eta_{l_1}^{(1)}$. We shall prove that within $G$ all the merged subpaths semi-reachable through $\phi$ by $\eta_1^{(1)}, \eta_2^{(1)}, \cdots, \mbox{ or } \eta_{l_1}^{(1)}$ will spread out to no less than $m$ $\psi$-paths (here by a set of merged subpaths $A$ spreading out to a set of paths $B$, we mean each element in $A$ is a subpath of some path in $B$, and every path in $B$ contain at least one element in $A$ as subpath). If $\eta_1^{(1)}, \eta_2^{(1)}, \cdots, \eta_{l_1}^{(1)}$ spread out to no less than $m$ $\psi$-paths, there is nothing to prove (since for each $i$ there must be at least one merged subpath immediately ahead of $\eta_i^{(1)}$ on some $\psi$-path). Now assume that $\eta_1^{(1)}, \eta_2^{(1)}, \cdots, \eta_{l_1}^{(1)}$ are confined within $m-1$ $\psi$-paths. Then one can prove that there is at least one $l_1^*$ such that within $\mathbb{R}$ there is a merged $\phi_{j_2}$-subpath ($j_2 \neq j_1$), say $\gamma^{(1)}$, is immediately ahead of $\eta_{l_1^*}^{(1)}$ on some $\psi$-path (and thus $\gamma^{(1)}$ and all the merged subpath larger than $\gamma^{(1)}$ on $\phi_{j_2}$ are semi-reachable by $\eta_{l_1^*}^{(1)}$), since otherwise, $l_1$ must be smaller than $\mathcal{M}(1, m-1)$, which means $\mathbb{R} \backslash \mathbb{R}|b(\eta_1^{(1)}), \cdots, b(\eta_{l_1}^{(1)}))$, consisting of subpaths from $m-1$ $\phi$-paths and $n$ $\psi$-paths, will have more than $\mathcal{M}(m-1, n)$ mergings, which implies reroutings can be done to reduce the number of mergings. Now pick a $l_1^*$ such that the corresponding $\eta_{l_1^*}^{(1)}$ is the smallest among such merged subpaths. Let $\mathbb{R}^{(1)}=\mathbb{R} \backslash \mathbb{R}|b(\eta_{l_1^*}^{(1)}))$. Within $\mathbb{R}|b(\eta_{l_1^*}^{(1)}))$, we can only have at most $\mathcal{M}(1, m-1)+1$ $\phi_{j_1}$-mergings and $\mathcal{M}(m-1, n)$ non-$\phi_{j_1}$-mergings, which implies that
$$
|\mathbb{R}|b(\eta_{l_1^*}^{(1)}))|_{\mathcal{M}} \leq \mathcal{M}(1, m-1)+1+\mathcal{M}(m-1, n).
$$
Now suppose we have $j_1, j_2, \cdots, j_{k+1}$ and $\mathbb{R}^{(k)}$ already, suppose within $\mathbb{R}^{(k)}$, $\alpha_{j_1}, \alpha_{j_2}, \cdots, \alpha_{j_{k+1}}$ merge with $\psi$ at $\eta_1^{(k+1)}, \eta_2^{(k+1)}, \cdots, \eta_{l_{k+1}}^{(k+1)}$. As argued above, without loss of generality, we can assume that these paths are confined within $m-1$ $\psi$-paths. Then one can prove that there is at least one $l_{k+1}^*$ such that there is merged $\phi_{j_{k+2}}$-subpath ($j_{k+2} \neq j_1, j_2, \cdots, j_{k+1}$), say $\gamma^{(k+1)}$, is immediately ahead of $\eta_{l_1^*}^{(k+1)}$ on some $\psi$-path (and thus $\gamma^{(k+1)}$ and all the merged subpaths larger than $\gamma^{(k+1)}$ on $\phi_{j_{k+2}}$ are semi-reachable through $\phi$ by $\eta_{l_{k+1}^*}^{(k+1)}$, thus semi-reachable through $\phi$ by some $\eta_i^{(1)}$), since otherwise, $l_{k+1}$ must be smaller than $\mathcal{M}(k+1, m-1)$, which means $\mathbb{R}^{(k)} \backslash \mathbb{R}^{(k)}|b(\eta_1^{(k+1)}), \cdots, b(\eta_{l_{k+1}}^{(k+1)}))$, consisting of subpaths from $m-k-1$ $\phi$-paths and $n$ $\psi$-paths, will have more than $\mathcal{M}(m-k-1, n)$ mergings, which implies reroutings can be done to reduce the number of mergings. Pick a $l_{k+1}^*$ such that $\eta_{l_{k+1}^*}^{(k+1)}$ is the smallest among such merged subpaths and define $\mathbb{R}^{(k+1)}=\mathbb{R}^{(k)} \backslash \mathbb{R}^{(k)}|b(\eta_{l_{k+1}^*}^{(k+1)}))$. Similarly one checks that within $\mathbb{R}^{(k)}|b(\eta_{l_{k+1}^*}^{(k+1)}))$, we can only have at most $\mathcal{M}(k+1, m-1)+1$ $\phi_{j}$-mergings ($j=j_1, j_2, \cdots, \mbox{ or } j_{k+1}$) and $\mathcal{M}(m-k-1, n)$ $\phi_{j}$-mergings (here $j \neq j_1, j_2, \cdots, \mbox{ and } j_{k+1}$), which implies that
$$
|\mathbb{R}^{(k)}|b(\eta_{l_{k+1}^*}^{(k+1)}))|_{\mathcal{M}} \leq \mathcal{M}(k+1, m-1)+1+\mathcal{M}(m-k-1, n).
$$
So eventually we will have $\mathbb{R}^{(m-1)}$, within which all the merged subpaths are semi-reachable through $\phi$ by $\eta_1^{(1)}, \eta_2^{(1)}, \cdots, \mbox{ or } \eta_{l_1}^{(1)}$. One checks that $|\mathbb{R}^{(m-1)}|_{\mathcal{M}} \geq \mathcal{M}(m, m-1)+1$, which implies that within $G$ all the merged subpaths semi-reachable through $\phi$ by $\eta_1^{(1)}, \eta_2^{(1)}, \cdots, \mbox{ or } \eta_{l_1}^{(1)}$ will spread out to no less than $m$ $\psi$-paths.

Now within $G$, consider all the merged subpaths semi-reachable through $\phi$ by $\eta_1^{(1)}, \eta_2^{(1)}$,  $\cdots, \mbox{ or } \eta_{l_1}^{(1)}$. As shown above, no less than $m$ $\psi$-paths, say $\psi_1, \psi_2, \cdots, \psi_{m'}$ ($m' \geq m$), contain at least one of the semi-reachable merged subpaths. Consider the smallest such merged  subpaths on each of $\psi_1, \psi_2, \cdots, \psi_{m'}$, say $\pi_1, \pi_2, \cdots, \pi_{m'}$, respectively. We can assume that none of $\pi_1, \pi_2, \cdots, \pi_{m'}$ belongs to $\phi_{j_1}$, otherwise for some $i$, $\eta_i^{(1)}$ is semi-reachable through $\phi$ by itself from above, then certain reroutings can be done to reduce the number of mergings. So, at least two of such smallest merged subpaths, say $\pi_i$, $\pi_j$, will belong to the same $\phi$-path. Assume that $\pi_i$ is smaller than $\pi_j$ on this $\phi$-path. If there is a merged subpath $\pi'_j$ immediately ahead of $\pi_j$ on $\psi_j$, then by definition, $\pi'_j$ will be semi-reachable by some $\eta_i^{(1)}$ as well, which contradicts the fact that $\pi_j$ is the smallest semi-reachable merged subpath on $\psi_j$. As a consequence, at least one of $\pi_1, \pi_2, \cdots, \pi_{m'}$ will be in fact the smallest merged subpath on the corresponding $\psi$-path. Without loss of generality, we assume that $\pi_1$ is in fact the smallest merged subpath on $\psi_1$ and $\pi_1$ is semi-reachable through $\phi$ from above by some $\eta_i^{(1)}$.

Apparently $\pi_1$ is within $\mathbb{S}$, since otherwise $\mathbb{S} \backslash \psi_1$, consisting of subpaths from $m$ $\phi$-paths and $n-1$ $\psi$-paths, will have more than $\mathcal{M}(m, n-1)$ mergings, which implies certain merging reducing reroutings can be done. With the same argument, we can assume
$$
|\mathbb{S}|a(\pi_1)) \backslash \psi_1|_{\mathcal{M}} \leq \mathcal{M}(m, n-1).
$$
As a consequence of this, we have
$$
|\mathbb{S}^{(0)}\stackrel{\triangle}{=}\mathbb{S} \backslash \mathbb{S}|a(\pi_1))|_{\mathcal{M}} \geq \sum_{j=1}^{m-2} (\mathcal{M}(j, n)+1+\mathcal{M}(m-j, n))+\mathcal{M}(m-1,n)+1.
$$
In the following, we shall prove that within $G$, some $\eta_i^{(1)}$ is semi-reachable through $\phi$ by itself from above. Now pick $i_1$ such that $\phi_{i_1}$ contains $\pi_1$ and assume that within $\mathbb{S}^{(0)}$, $\phi_{i_1}$ merges with $\psi$ at the merged subpaths $\zeta_1^{(1)}=\pi_1, \zeta_2^{(1)}, \cdots, \zeta_{r_1}^{(1)}$. Then one can prove that there is at least one $r_1^*$ such that there is a merged $\phi_{i_2}$-subpath ($i_2 \neq i_1$), say $\lambda^{(1)}$, is immediately ahead of $\zeta_{r_1^*}^{(1)}$ on some $\psi$-path (and thus $\lambda^{(1)}$ and all the merged subpath larger than $\lambda^{(1)}$ on $\phi_{i_2}$ are semi-reachable through $\phi$ by $\zeta_{r_1^*}^{(1)}$, and thus semi-reachable through $\phi$ by some $\eta_i^{(1)}$), since otherwise, $r_1$ must be smaller than $\mathcal{M}(1, n)$, which means $\mathbb{S}^{(0)} \backslash \mathbb{S}^{(0)}|b(\zeta_1^{(1)}), \cdots, b(\zeta_{l_1}^{(1)}))$, consisting of suppaths from $m-1$ $\phi$-paths and $n$ $\psi$-paths, will have more than $\mathcal{M}(m-1, n)$ mergings, which implies reroutings can be done to reduce the number of mergings. Now pick a $r_1^*$ such that the corresponding $\zeta_{r_1^*}^{(1)}$ is the smallest among such merged subpaths and let $\mathbb{S}^{(1)}=\mathbb{S}^{(0)} \backslash \mathbb{S}^{(0)}|b(\zeta_{r_1^*}^{(1)}))$. Within $\mathbb{S}^{(0)}|b(\zeta_{r_1^*}^{(1)}))$, we can only have at most $\mathcal{M}(1, n)+1$ $\phi_{i_1}$-mergings and at most $\mathcal{M}(m-1, n)$ non-$\phi_{i_1}$-mergings, which implies that
$$
|\mathbb{S}^{(0)}|b(\zeta_{l_1^*}^{(1)}))|_{\mathcal{M}} \leq \mathcal{M}(1, n)+1+\mathcal{M}(m-1, n).
$$
Now suppose we have $j_1, j_2, \cdots, j_{k+1}$ and $\mathbb{S}^{(k)}$ already, and suppose within $\mathbb{S}^{(k)}$, $\phi_{j_1}, \phi_{j_2}, \cdots, \phi_{j_{k+1}}$ merge with $\psi$ at $\zeta_1^{(k+1)}, \zeta_2^{(k+1)}, \cdots, \zeta_{r_{k+1}}^{(k+1)}$. Then one can prove that there is at least one $r_{k+1}^*$ such that there is $\phi_{i_{k+2}}$-merged subpath ($i_{k+2} \neq i_1, i_2, \cdots, i_{k+1}$), say $\lambda^{(k+1)}$, is immediately ahead of $\zeta_{r_{k+1}^*}^{(k+1)}$ on some $\psi$-path (and thus $\lambda^{(k+1)}$ and all the merged subpath larger than $\lambda^{(k+1)}$ on $\phi_{i_{k+2}}$ are semi-reachable through $\phi$ by $\zeta_{r_{k+1}^*}^{(k+1)}$, thus semi-reachable through $\phi$ by some $\eta_i^{(1)}$), since otherwise, $r_{k+1}$ must be smaller than $\mathcal{M}(k+1, n)$, which means $\mathbb{S}^{(k)} \backslash \mathbb{S}^{(k)}|b(\zeta_1^{(k+1)}), \cdots, b(\zeta_{r_{k+1}}^{(k+1)}))$, consisting of subpaths from $m-k-1$ $\phi$-paths and $n$ $\psi$-paths, will have more than $\mathcal{M}(m-k-1, n)$ mergings, which implies reroutings can be done to reduce the number of mergings. Pick a $r_{k+1}^*$ such that $\zeta_{r_{k+1}^*}^{(k+1)}$ is the smallest such merged subpath and define $\mathbb{S}^{(k+1)}=\mathbb{S}^{(k)} \backslash \mathbb{S}^{(k)}|b(\zeta_{r_{k+1}^*}^{(k+1)}))$. Similarly one checks that within $\mathbb{S}^{(k)}|b(\zeta_{r_{k+1}^*}^{(k+1)}))$, we can only have at most $\mathcal{M}(k+1, n)+1$ $\phi_{i}$-mergings ($i=i_1, i_2, \cdots, \mbox{ or } i_{k+1}$) and $\mathcal{M}(m-k-1, n)$ $\alpha_{i}$-mergings ($i \neq i_1, i_2, \cdots, \mbox{ and } i_{k+1}$), which implies that
$$
|\mathbb{S}^{(k)}|b(\zeta_{r_{k+1}^*}^{(k+1)}))|_{\mathcal{M}} \leq \mathcal{M}(k+1, n)+1+\mathcal{M}(m-k-1, n).
$$
So eventually we will have $\mathbb{S}^{(m-2)}$. One checks that
$$
|\mathbb{S}^{(m-2)}|_{\mathcal{M}} \geq \mathcal{M}(m-1, n)+1,
$$
which implies all $\phi_i$'s has merged subpaths within $\mathbb{S}$ semi-reachable through $\phi$ by some $\eta_i^{(1)}$. In particular, some merged $\phi_{j_1}$-subpath within $\mathbb{S}$ is semi-reachable by some $\eta_i^{(1)}$, thus some $\eta_i^{(1)}$ is semi-reachable by itself from above, so certain reroutings can be done to reduce the number of mergings. With this contradiction, we establish the proposition.

\end{proof}

\begin{rem}
Define $w_i=\sum_{j=1}^i (\mathcal{M}(j, m-1)+1)$. Note that Proposition~\ref{ILikeIt} is still true if $U(m, n)$ is replaced by $m w_m$, which produces an alternative upper bound on $\mathcal{M}(m, n)$. One can obtain the proof of this by replacing $U(m, n)$ in the first and second paragraphs in the proof of Proposition~\ref{ILikeIt} with $mw_m$ and replacing the third paragraph in the proof of Proposition~\ref{ILikeIt} with the following paragraph.

Now assume that we find $\eps_1, \eps_2, \cdots, \eps_n$ such that
$$
\mathbb{S}\stackrel{\triangle}{=}G|b(\eps_1), \cdots, b(\eps_n))
$$
has no less than $V(m, n)$ mergings and $\mathbb{R}=G \backslash \mathbb{S}$ has no less than $m w_m$ mergings. By the Pigeonhole principle, there must be at least one $\phi_j$ such that $\phi_j$ merge with $\psi$ for no less than $w_m$ times. Without loss of generality, assume that $\phi_{j_1}$ merges with $\psi$ subsequently at $\eta_1^{(1)}, \eta_2^{(1)}, \cdots, \eta_{l_1}^{(1)}$, here $l_1 \geq w_m$. Now within $\mathbb{R}|b(\eta_{w_1}^{(1)}))$, unless $\phi_{j_1}$ merges with no less than $m$ $\psi$-paths, there exists $j_2 \neq j_1$ such that a merged $\phi_{j_2}$-subpath, say $\gamma^{(1)}$, is immediately ahead of certain merged $\phi_{j_1}$-subpath, say $\eta_{l_1^*}^{(1)}$. So $\gamma^{(1)}$ and any merged subpath larger than $\gamma^{(1)}$ on $\phi_{j_2}$ is semi-reachbable through $\phi$ by $\eta_{l_1^*}^{(1)}$. Now continue the argument inductively and suppose we have already obtained $j_1, j_2, \cdots, j_{k+1}$. Then within $\mathbb{R}|b(\eta_{w_{k+1}}^{(1)})) \backslash \mathbb{R}|b(\eta_{w_{k}}^{(1)}))$, assume that $\phi_{j_1}, \phi_{j_2}, \cdots, \phi_{j_{k+1}}$ merge with $\psi$ at $\eta_1^{(k+1)}, \eta_2^{(k+1)}, \cdots, \eta_{l_{k+1}}^{(k+1)}$, here obviously $l_{k+1} \geq w_{k+1}-w_k$. Unless $\phi_{j_1}, \phi_{j_2}, \cdots, \phi_{j_{k+1}}$ merge with no less than $m$ $\psi$-paths within $\mathbb{R}|b(\eta_{w_{k+1}}^{(1)})) \backslash \mathbb{R}|b(\eta_{w_{k}}^{(1)}))$, there exists $j_{k+2} \neq j_1, j_2, \cdots, j_{k+1}$ such that a merged $\phi_{j_{k+2}}$-subpath, say $\gamma^{(k+1)}$, is immediately ahead of some $\eta_{l_{k+1}^*}^{(k+1)}$. Thus $\gamma^{(k+1)}$ and any merged subpaths larger than $\gamma^{(k+1)}$ on $\phi_{j_{k+2}}$ are semi-reachable through $\phi$ by $\eta_{l_{k+1}^*}^{(k+1)}$, and thus by some $\eta_i^{(1)}$. Eventually one can show that within $\mathbb{R}|b(\eta_{w_{m}}^{(1)})) \backslash \mathbb{R}|b(\eta_{w_{m-1}}^{(1)}))$, all merged non-$\phi_{j_1}$-subpaths are semi-reachable through $\phi$ by some $\eta_i^{(1)}$. Since
$$
|\mathbb{R}|b(\eta_{w_{m}}^{(1)})) \backslash \mathbb{R}|b(\eta_{w_{m-1}}^{(1)}))|_{\mathcal{M}} \geq \mathcal{M}(m-1, m)+1,
$$
all merged subpaths within $\mathbb{R}|b(\eta_{w_{m}}^{(1)})) \backslash \mathbb{R}|b(\eta_{w_{m-1}}^{(1)}))$ spread out to no less than $m$ $\psi$-paths, which implies that within $G$ all the merged subpaths semi-reachable through $\phi$ by $\eta_1^{(1)}, \eta_2^{(1)}$, $\cdots, \mbox{ or } \eta_{l_1}^{(1)}$ will spread out to no less than $m$ $\psi$-paths.

\end{rem}

\begin{exmp} \label{MWithOne}
It was first shown in~\cite{Ta2003} that $\mathcal{M}(1, n)=n$. To see this, consider any acyclic directed graph $G(E, V)$ with $2$ distinct sources $S_1, S_2$ and $2$ distinct sinks $R_1, R_2$, where the min-cut between $S_i$ and $R_i$ is denoted by $c_i$; here $c_1=1$ and $c_2=n$. Pick a set of Menger's path $\alpha_i=\{\alpha_{i, 1}, \alpha_{i, 2}, \cdots, \alpha_{i, c_i}\}$ from $S_i$ to $R_i$. If $\alpha_{1, 1}$ merges with some $\alpha_2$-path, say $\alpha_{2, j}$, at least twice, say at $e$ and $f$. Then we can replace $\alpha_{1, 1}[a(e), a(f)]$, the subpath of $\alpha_{1, 1}$ starting from $a(e)$ to $a(f)$, by $\alpha_{2, j}[a(e), a(f)]$, the subpath of $\alpha_{2, j}$ starting from $a(e)$ to $a(f)$. After this rerouting, the new $\alpha_{1, 1}$ has fewer mergings with $\alpha_2$. This shows that
$$
\mathcal{M}(1, n) \leq n,
$$
since $\alpha_{1, 1}$ can be chosen to merge with each $\alpha_2$-path for at most once. For the other direction, by Proposition~\ref{Isolated}, we have
$$
\mathcal{M}(1, n) \geq \sum_{i=1}^n \mathcal{M}(1, 1)=n,
$$
the last equality follows from the simple fact that $\mathcal{M}(1, 1)=1$.
\end{exmp}

\begin{rem}
Note that Example~\ref{MWithOne} together with the inductive argument in the proof of Proposition~\ref{ILikeIt} gives
an alternative proof of that $\mathcal{M}(c_1, c_2)$ is finite.
\end{rem}

\begin{pr} \label{k-2}
$$
\mathcal{M}(2, n)=3n-1.
$$
\end{pr}

\begin{proof}

We first show that $\mathcal{M}(2, n)$ is upper bounded by $3n-1$.

Consider an acyclic graph $G$ with two distinct sources $S_1, S_2$ and two distinct sinks $R_1, R_2$. Assume that a group of Menger's paths $\psi=\{\psi_1, \psi_2\}$ connect $S_1, R_1$, and another group of Menger's paths $\phi=\{\phi_1, \phi_2, \cdots, \phi_n\}$ connect $S_2, R_2$;
and assume that $G$ consists of only the above-mentioned two groups of Menger's paths and $G$ is non-reroutable with respect to $\phi$ and $\psi$.

Assume that out of $n$ $\phi$-paths, there are $k$ $\phi$-paths, say $\phi_1, \phi_2, \cdots, \phi_k$, each of which merges for at least $3$ times. Notice that when $k=0$, the total number of mergings in $G$ is upper bounded by $2n$; so in the following, we only consider the case when $k \geq 1$. For $i=1, 2, \cdots, k$, assume that $\phi_i$ sequentially merges at $\gamma_{i, 1}, \gamma_{i, 2}, \cdots, \gamma_{i, n_i}$. Let $\ell(i, j)$ denote the index of the $\psi$-path which $\gamma_{i, j}$ belongs to. Since each $\phi$-path has to merge with $\psi_1, \psi_2$ alternately, we have
$\ell(i, j)=\ell(i, k)$ if $j=k \mod 2$.

Note that for each pair $\gamma_{i, j}, \gamma_{i, j+2}$, there must exist one merged subpath, say $\eta_{i, j}$, which is in between $\gamma_{i, j}$ and $\gamma_{i, j+2}$ on $\psi_{\ell(i, j)}$. One easily checks that for any $i, j$, the $\phi$-path which $\eta_{i, j}$ belongs to can merge with $\psi$-paths at most twice. We then claim that, for fixed $i$, one can choose all $\eta_{i, j}$ such that each $\eta_{i, j}$ belongs to a different $\phi$-path. This can be shown by an inductive approach on the length of path $\phi_i$. The case when $n_i=3$ is trivial. Now suppose the claim is established for $n_i=3, 4, \cdots, l$. We next show that the claim is also true for $n_i=l+1$. First note that if $\eta_{i, j}$ and $\eta_{i, k}$ ($j < k$) share the same $\phi$-path, then necessarily $k=j+1$. Now consider $\eta_{i, 1}$. If $\eta_{i, 1}$ does not share the same $\phi$-path with $\eta_{i, 2}$, then by induction assumptions, the claim, when restricted to $\phi_i[a(\gamma_{i, 2}), b(\gamma_{i, n_i})]$, is true, thus implying the claim without any restriction is true. Hence, in the following, we only consider the case when $\eta_{i, 1}$ and $\eta_{i, 2}$ share the same $\phi$-path. For this case, there must be a merged subpath either in between $\gamma_{i, 1}$ and $\eta_{i, 1}$ on $\psi_{\ell(i, 1)}$ or in between $\eta_{i, 2}$ and $\gamma_{i, 4}$ on $\psi_{\ell(i, 2)}$, since, otherwise, $\gamma_{i, 4}$ would be semi-reachable through $\phi$ via $\gamma_{i, 4}, \eta_{i, 2}, \gamma_{i, 1}, \eta_{i, 1}, \gamma_{i, 4}$ from above by itself, which implies that $G$ is reroutable. If there is a merged subpath, say $\eta'_{i, 1}$, in between $\gamma_{i, 1}$ and $\eta_{i, 1}$ on $\phi_{\ell(i, 1)}$, then one can reset $\eta_{i, 1}$ to be $\eta'_{i, 1}$, then apply induction to $\phi_i[a(\gamma_{i, 3}), b(\gamma_{i, n_i})]$ to establish the claim. Hence, in the following, we further assume that there is no merged subpath in between $\gamma_{i, 1}$ and $\gamma_{i, 3}$ on $\psi_{\ell(i, 1)}$, thus there must exist a merged subpath, say $\eta'_{i, 2}$, in between $\eta_{i, 2}$ and $\gamma_{i, 4}$ on $\psi_{\ell(i, 2)}$. If $\eta'_{i, 2}$ does not share the same $\phi$-path with $\eta_{i, 3}$, we can reset $\eta_{i, 2}$ to be $\eta'_{i, 2}$ and apply induction on $\phi_i[a(\gamma_{i, 3}), b(\gamma_{i, n_i})]$ to establish the claim. Hence in the following, we further assume $\eta'_{i, 2}$ does share the same $\phi$-path with $\eta_{i, 3}$. For $j=2, 3, \cdots, n_i-3$, we say $\eta_{i, j}$ is {\em type I} if there exists exactly one merged subpath $\eta'_{i, j}$ in between $\eta_{i, j}$ and $\gamma_{i, j+2}$, and $\eta'_{i, j}$ and $\eta_{i, j+1}$ share the same $\phi$-path. Let $2 \leq k \leq n_i-3$ be the smallest index such that $\eta_{i, k}$ is not type I, meaning either (there is no merged subpath in between $\eta_{i, k}$ and $\gamma_{i, k+2}$ on $\psi_{\ell(i, k+2)}$) or (there is a merged subpath, say $\eta'_{i, k}$, in between $\eta_{i, k}$ and $\gamma_{i, k+2}$ on $\psi_{\ell(i, k+2)}$, however $\eta'_{i, k}$ does not share the same $\phi$-path with $\eta_{i, k+1}$). The former case implies that $\gamma_{i, k+2}$ is semi-reachable through $\phi$ via $\gamma_{i, k+2}, \eta_{i, k}, \eta'_{i, k-1}, \eta_{i, k-1}, \cdots, \eta_{i, 1}, \gamma_{i, 1}, \gamma_{i, k+2}$ from above by itself, thus it would not occur; while for the latter case, one can reset $\eta_{i, k}$ to be $\eta'_{i, k}$, $\eta_{i, k-1}$ to be $\eta'_{i, k-1}$, $\cdots$, $\eta_{i, 2}$ to be $\eta'_{i, 2}$, and apply induction on $\phi_i[a(\gamma_{i, k+2}), b(\gamma_{i, n_i})]$ to establish the claim. So in the following, we further assume that all $\eta_{i, j}$, $j=2, 3, \cdots, n_i-3$, are type I. Now consider $\eta_{i, n_i-2}$. One checks that there must exist a merged subpath, say $\eta'_{i, n_i-2}$ in between $\eta_{i, n_i-2}$ and $\gamma_{i, n_i}$ on $\psi_{\ell(i, n_i-2)}$, since otherwise, again, $\gamma_{i, n_i}$ would be semi-reachable through $\phi$ via $\gamma_{i, n_i}, \eta_{i, n_i-2}, \eta'_{i, n_i-3}, \eta_{i, n_i-3}, \cdots, \eta_{i, 1}, \gamma_{i, 1}, \gamma_{i, n_i}$ from above itself. Thus we can reset $\eta_{i, n_i}$ to be $\eta'_{i, n_i}$, $\eta_{i, n_i-1}$ to be $\eta'_{i, n_i-1}$, $\cdots$, $\eta_{i, 2}$ to be $\eta'_{i, 2}$. One checks that each of newly defined $\eta_{i, j}$ belong to different $\phi$-path.

One also verifies that for any $i, j=1, 2, \cdots, k$, $\phi_i$ and $\phi_j$ are ``well-separated''; more precisely,
one of the pair, say $\phi_i$, must be ``smaller'' than the other one, $\phi_j$, in the sense that the merged subpaths by $\phi_i$ on $\psi_1, \psi_2$ must be smaller than the merged subpaths by $\phi_j$ on $\psi_1, \psi_2$, respectively. Through renumbering, if necessary, we assume that for any $1 \leq i < j \leq k$, $\phi_i$ is always smaller than $\phi_j$. Then with this, one checks that for any $1 \leq i_1 < i_2 \leq k$, $\eta_{i_1, j_1}$ and $\eta_{i_2, j_2}$ share the same $\phi$-path if and only if $i_2=i_1+1$ and $j_1=n_{i_1}-2, j_2=1$. Thus
there must exist at least $(n_1-2+n_2-2+\cdots+n_k-2)-(k+1)$ $\phi$-paths, each of which contains some $\eta_{i, j}$ as subpath, and again, each of these $\phi$-paths can merge at most twice.

So the total number of mergings in $G$ is upper bounded by
$$
n_1+n_2+\cdots+n_k+2(n-k),
$$
subject to $(n_1-2)+(n_2-2)+\cdots+(n_k-2)-(k-1) \leq n-k$ (the number of $\phi$-paths that contains some $\eta_{i, j}$ as subpath is lower bounded by $n-k$). One then checks that the number of mergings is upper bounded by $3n-1$, thus we conclude that $\mathcal{M}(2, n)$ is upper bounded by $3n-1$.

To show $\mathcal{M}(2, n)$ is also lower bounded by $3n-1$, it suffices to construct a non-reroutable graph $G$ with $M(G)=3n-1$. For instance, we can first choose $\phi_1$ to alternately merge with $\psi_1, \psi_2$ $n+1$ times at $\gamma_1, \gamma_2, \cdots, \gamma_{n+1}$. Next we choose each $\phi_i$, $i=2, 3, \cdots, n$ to merge exactly twice, while ensuring that, for all $i < j$, $\phi$ is smaller than $\phi_j$ in the sense that the merged subpaths by $\phi_i$ on $\psi_1, \psi_2$ are smaller than the merged subpaths by $\phi_j$ on $\psi_1, \psi_2$, respectively. Moreover we also require that $\phi_{2i}$ first merges with $\psi_1$ in between $\gamma_{2i-1}$ and $\gamma_{2i+1}$, and then merge with $\psi_2$ in between $\gamma_{2i-2}$ and $\gamma_{2i}$, and that $\phi_{2i+1}$ first merges with $\psi_2$ in between $\gamma_{2i}$ and $\gamma_{2i+2}$, and then merges with $\psi_1$ in between $\gamma_{2i-1}$ and $\gamma_{2i+1}$ (see an example graph in Figure~\ref{2n-example} for the case $n=3$). It can be checked that such a graph is non-reroutable and the number of mergings is $3n-1$.

\begin{figure}
\psfrag{S1}{$S_1$} \psfrag{S2}{$S_2$} \psfrag{R1}{$R_1$} \psfrag{R2}{$R_2$}
\psfrag{a}{$\phi_1$} \psfrag{b}{$\phi_2$} \psfrag{c}{$\phi_3$} \psfrag{psi1}{$\psi_1$} \psfrag{psi2}{$\psi_2$} \psfrag{psi3}{$\psi_3$}
\psfrag{g1}{$\gamma_1$} \psfrag{g2}{$\gamma_2$} \psfrag{g3}{$\gamma_3$}

\centerline{\includegraphics[width=2.5in]{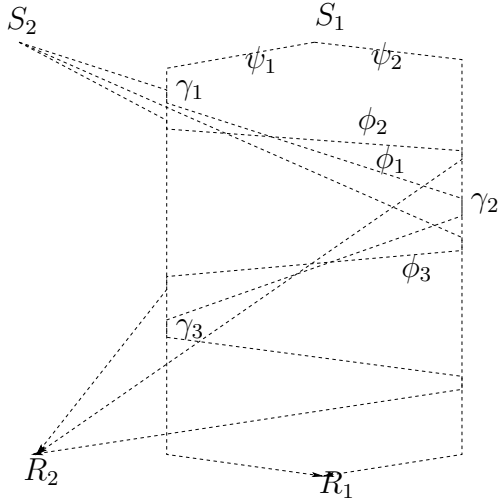}}
\caption{an example graph achieving $\mathcal{M}(2, 3)$}
\label{2n-example}
\end{figure}

\end{proof}

We next prove that when fixing $c_1$, $\mathcal{M}(c_1, c_2)$ grows at most linearly with respect to $c_2$.
\begin{pr} \label{k-n}
For any fixed $k$, there exists a positive constant $C_k$ such that for all $n$,
$$
\mathcal{M}(k, n) \leq C_k n.
$$
\end{pr}

\begin{proof}

We proceed by induction on $k$. It follows from $\mathcal{M}(1, n)=n$ (see Example~\ref{MWithOne}) that for the case when $k=1$, the theorem is true with $C_1=1$. Now for any $k \geq 2$, assume that for any $i=1, 2, \cdots, k-1$, there exists a positive constant $C_i$ such that for all $n$,
$$
\mathcal{M}(i, n) \leq C_i n;
$$
we next show that there exists a positive constant $C_k$ such that for all $n$,
$$
\mathcal{M}(k, n) \leq C_k n.
$$

Again for an acyclic graph $G$ with two distinct sources $S_1, S_2$ and two distinct sinks $R_1, R_2$, let $\psi=\{\psi_1, \psi_2, \cdots, \psi_k\}$ denote a group of Menger's paths connecting $S_1, R_1$, and let $\phi=\{\phi_1, \phi_2, \cdots, \phi_n\}$ denote another group of Menger's paths connecting $S_2, R_2$. We further assume that $G$ consists of only the above-mentioned two groups of Menger's paths and $G$ is non-reroutable with respect to $\phi$ and $\psi$. We only need to prove that there exists $C_k$ such that $|G|_{\mathcal{M}} \leq C_k n$, where $|G|_{\mathcal{M}}$ denotes the number of mergings in $G$ (Note that when $G$ is non-reroutable, $|G|_{\mathcal{M}}$ is in fact $M(G)$).

Consider the following iterative procedure, where, for notational simplicity, we treat a graph as the union of its vertex set and edge set. Initially set $\mathbb{S}^{(0)}=\emptyset$, and $\mathbb{R}^{(0)}=G$. Now for an arbitrary yet fixed $K > 0$ (we shall choose $K$ large enough later) and each $j=1, 2, \cdots, k$, pick merged subpaths $\gamma_{0, j}$   such that $\gamma_{0, j}$ belongs to path $\psi_j$ and
$$
|\mathbb{R}^{(0)}|b(\gamma_{0, 1}), b(\gamma_{0, 2}), \cdots, b(\gamma_{0, k}))|_{\mathcal{M}}=K;
$$
note that, without loss of generality, we can assume that $\gamma_{0, j}$ is the largest merged subpath from $\mathbb{R}^{(0)}|b(\gamma_{0, 1}), b(\gamma_{0, 2}), \cdots, b(\gamma_{0, k}))$ on $\psi_j$ (one can choose $\gamma_{0, j}$ to be $S_1$ if such merged subpath does not exist on $\psi_j$).
Now set
$$
\mathbb{S}^{(1)}=\mathbb{S}^{(0)} \cup \mathbb{R}^{(0)}|b(\gamma_{0, 1}), b(\gamma_{0, 2}), \cdots, b(\gamma_{0, k}))
$$
and
$$
\mathbb{R}^{(1)}=\mathbb{R}^{(0)} \backslash \mathbb{R}^{(0)}|b(\gamma_{0, 1}), b(\gamma_{0, 2}), \cdots, b(\gamma_{0, k})).
$$
If a merged subpath is the smallest or the largest one on a $\phi$-path, we say it is a {\em terminal} merged subpath on the $\phi$-path, or simply a $\phi$-terminal merged subpath. Now suppose that we already obtain
$$
\mathbb{S}^{(i)}=\mathbb{S}^{(i-1)} \cup \mathbb{R}^{(i-1)}|b(\gamma_{i-1, 1}), b(\gamma_{i-1, 2}), \cdots, b(\gamma_{i-1, k}))
$$
and
$$
\mathbb{R}^{(i)}=\mathbb{R}^{(i-1)} \backslash \mathbb{R}^{(i-1)}|b(\gamma_{i-1, 1}), b(\gamma_{i-1, 2}), \cdots, b(\gamma_{i-1, k})),
$$
where $\mathbb{R}^{(i-1)}|b(\gamma_{i-1, 1}), b(\gamma_{i-1, 2}), \cdots, b(\gamma_{i-1, k}))$ contains exactly $K$ mergings and at least one $\phi$-terminal merged subpath, we then continue to pick merged subpath $\gamma_{i, j}$ on $\psi_j$ from $R^{(i)}$ such that
$$
|\mathbb{R}^{(i)}|b(\gamma_{i, 1}), b(\gamma_{i, 2}), \cdots, b(\gamma_{i, k}))|_{\mathcal{M}}=K,
$$
here, again, each $\gamma_{i-1, j}$, $j=1, 2, \cdots, k$, is chosen to be largest merged subpath on $\psi_j$; and we set
$$
\mathbb{S}^{(i+1)}=\mathbb{S}^{(i)} \cup \mathbb{R}^{(i)}|b(\gamma_{i, 1}), b(\gamma_{i, 2}), \cdots, b(\gamma_{i, k})),
$$
and
$$
\mathbb{R}^{(i+1)}=\mathbb{R}^{(i)} \backslash \mathbb{R}^{(i)}|b(\gamma_{i, 1}), b(\gamma_{i, 2}), \cdots, b(\gamma_{i, k})).
$$
We will further continue in this fashion, if necessary, to obtain $\mathbb{S}^{(2)}, \mathbb{R}^{(2)}, \mathbb{S}^{(3)}, \mathbb{R}^{(3)}, \cdots$ until we obtain $\mathbb{S}^{(i_0)}, \mathbb{R}^{(i_0)}$ such that $\mathbb{S}^{(i_0)} \backslash \mathbb{S}^{(i_0-1)}$ does not contain any $\phi$-terminal merged subpaths. Note that each $\mathbb{S}^{(j)} \backslash \mathbb{S}^{(j-1)}$ ($j=1, 2, \cdots, i$) has $K$ mergings, we thus call each of them a {\em $K$-trunk}. The first $i_0-1$ $K$-trunk, $\mathbb{S}^{(j)} \backslash \mathbb{S}^{(j-1)}$ ($j=1, 2, \cdots, i_0-1$) contains some $\phi$-terminal merged subpaths, we thus call these $K$-trunks {\em singular}; on the other hand, the $i_0$-th $K$-trunk does not contain any terminal merged subpaths on any $\phi$-path, we then call this $K$-trunk {\em normal}.

By Theorem~\ref{main} and the fact that $G$ is non-reroutable, we now choose $K$ so large that the number of critical merged subpaths within $\mathbb{S}^{(i_0)} \backslash \mathbb{S}^{(i_0-1)}$ is larger than $k$ (thus the number of critical merged subpaths within $\mathbb{S}^{(i_0)}$ is larger than $k$), here we say a merged subpath is {\em critical} within a subgraph of $G$ if the associated $\phi$-path, after merging at this merged subpath, does not merge anymore within this subgraph (note that since $\mathbb{S}^{(i_0)} \backslash \mathbb{S}^{(i_0-1)}$ is normal, the $\phi$-path will continue to merge within $\mathbb{R}^{(i_0)}$).

Now, let $T_{i_0}$ denote the set of all the merged subpaths within $\mathbb{R}^{(i_0)}$ which can semi-reach some critical merged subpath within $\mathbb{S}^{(i_0)}$ through group $\psi$ from below. One checks at least one of those $\psi$-paths, each of which contains at least one critical merged subpath within $\mathbb{S}^{(i_0)}$, does not contain any merged subpath within $T_{i_0}$ (since, otherwise, by a usual back-tracing argument, one checks that there exists a merging reducing rerouting). Assume that $\xi_{i_0, 1}, \xi_{i_0, 2}, \cdots, \xi_{i_0, m_{i_0}}$ ($1 \leq m_{i_0} \leq k-1$) are the largest merged subpaths from $T_{i_0}$, and they belong to paths $\psi_{j_{i_0, 1}}, \psi_{j_{i_0, 2}}, \cdots, \psi_{j_{i_0, m_{i_0}}}$, respectively. Let
$$
\bar{T}_{i_0}=\bigcup_{j=1}^{m_{i_0}} \psi_{j_{i_0, j}}[b(\gamma_{i_0-1, j_{i_0, j}}),  b(\xi_{i_0, j})],
$$
here, one can check that $\psi_{j_{i_0, j}}[b(\gamma_{i_0-1, j_{i_0, j}}),  b(\xi_{i_0, j})]$ is the ``segment'' of $\psi_{j_{i_0, j}}$ that is within $\mathbb{R}^{(i_0)}$ and before $b(\xi_{i_0, j})$, or more formally,
$$
\psi_{j_{i_0, j}}[b(\gamma_{i_0-1, j_{i_0, j}}),  b(\xi_{i_0, j})]=\psi_{j_{i_0, j}}[S_1, b(\xi_{i_0, j})] \cap \mathbb{R}^{(i_0)}.
$$

Note that for any $\xi_{i_0, j}$, $j=1, 2, \cdots, m_{i_0}$, the asociated $\phi$-path, from $\xi_{i_0, j}$, may merge outside $\bar{T}_{i_0}$ next time; if this $\phi$-path merge within $\bar{T}_{i_0}$ again after a number of mergings outside $\bar{T}_{i_0}$, we call it an {\em excursive} $\phi$-path (with respect to $\xi_{i_0, j}$). One checks that there are at most $k-2$ excursive $\phi$-paths (since, otherwise, we can find a cycle in $G$, which is a contradiction). So, letting $L_{i_0}$ denote the number of $\phi$-paths that contains at least one merged subpath within $\mathbb{R}^{(i_0)}|b(\xi_{i_0, 1}), b(\xi_{i_0, 2}), \cdots, b(\xi_{i_0, m_{i_0}}))$, the number of connected $\phi$-paths is upper bounded by $L_{i_0}+(k-2)$. Then, by induction assumptions,
$$
|\mathbb{R}^{(i_0)}|b(\xi_{i_0, 1}), b(\xi_{i_0, 2}), \cdots, b(\xi_{i_0, m_{i_0}})) \cap \bar{T}_{i_0}|_{\mathcal{M}} \leq C_{m_{i_0}} (L_{i_0}+(k-2)) \leq C_{k-1} L_{i_0}+C_{k-1} (k-2).
$$
On the other hand, for any merged subpath, say $\eta$, from $\bar{T}_{i_0}$ other than $\xi_{i_0, j}$, $j=1, 2, \cdots, m_{i_0}$, the associated $\phi$-path, from $\eta$, can only merge within $\bar{T}_{i_0}$ (note that this implies that, from $\mathbb{R}^{(i_0)}|b(\xi_{i_0, 1}), b(\xi_{i_0, 2}), \cdots, b(\xi_{i_0, m_{i_0}}))$, at most $k-1$ $\phi$-paths can merge further; this fact will be used later in the proof). One checks that there exists at least one $\psi_{j_{i_0, j}}$, $j=1, 2, \cdots, m_{i_0}$, which does not merge with any $\phi$-paths within $\mathbb{R}^{(i_0)}|b(\xi_{i_0, 1}), b(\xi_{i_0, 2}), \cdots, b(\xi_{i_0, m_{i_0}})) \backslash \bar{T}_{i_0}$ (again, since, otherwise, we can find a cycle in $G$, which is a contradiction). Thus, by induction assumptions,
$$
|\mathbb{R}^{(i_0)}|b(\xi_{i_0, 1}), b(\xi_{i_0, 2}), \cdots, b(\xi_{i_0, m_{i_0}})) \backslash \bar{T}_{i_0} |_{\mathcal{M}} \leq C_{k-1} L_{i_0},
$$
It then immediately follows that
$$
|\mathbb{R}^{(i_0)}|b(\xi_{i_0, 1}), b(\xi_{i_0, 2}), \cdots, b(\xi_{i_0, m_{i_0}}))|_{\mathcal{M}} \leq 2 C_{k-1} L_{i_0}+ C_{k-1} (k-2).
$$

Now set
$$
\mathbb{S}^{(i_0+1)}=\mathbb{S}^{(i_0)} \cup \mathbb{R}^{(i_0)}|b(\xi_{i_0, 1}), b(\xi_{i_0, 2}), \cdots, b(\xi_{i_0, m_{i_0}}))
$$
and
$$
\mathbb{R}_{(i_0+1)}=\mathbb{R}^{(i_0)} \backslash \mathbb{R}^{(i_0)}|b(\xi_{i_0, 1}), b(\xi_{i_0, 2}), \cdots, b(\xi_{i_0, m_{i_0}})).
$$

We roughly summarize what we have done so far. Roughly speaking, from the source side of graph $G$, we keep ``cutting'' $K$-trunks, $\mathbb{S}^{(i)} \backslash \mathbb{S}^{(i-1)}$, $i=1, 2, \cdots, i_0-1$, from $G$, until we obtain a normal $K$-trunk, $\mathbb{S}^{(i_0)} \backslash \mathbb{S}^{(i_0-1)}$, then, by cutting all the merged subpaths smaller than some of those merged subpaths that can semi-reach some critical merged subpaths within $\mathbb{S}^{(i_0)}$ through $\psi$ from below, we obtain a $\tilde{K}$-trunk, $\mathbb{S}^{(i_0+1)} \backslash \mathbb{S}^{(i_0)}$.

Similar operations can be done to $\mathbb{R}^{(i_0+1)}$. More precisely, we keep cutting $K$-trunks from $\mathbb{R}^{(i_0+1)}$ until we obtain a normal $K$-trunk $\mathbb{S}^{(i_1)} \backslash \mathbb{S}^{(i_1-1)}$, then we cut all the merged subpaths smaller than some of those merged subpaths that can semi-reach some critical merged subpaths within $\mathbb{S}^{(i_1)}$ to obtain a $\tilde{K}$-trunk, $\mathbb{S}^{(i_1+1)} \backslash \mathbb{S}^{(i_1)}$ with
$$
|\mathbb{S}^{(i_1+1)} \backslash \mathbb{S}^{(i_1)}|_{\mathcal{M}} \leq 2C_{k-1} L_{i_1}+C_{k-1} (k-2),
$$
where $L_{i_1}$ denotes the number of $\phi$-paths that contains at least one merged subpath within $\mathbb{S}^{(i_1+1)} \backslash \mathbb{S}^{(i_1)}$. We continue these operations in an iterative fashion to further obtain normal $K$-trunks and $\tilde{K}$-trunks (here, again, we are following the same notational convention as before),
$$
\mathbb{S}^{(i_1+1)} \backslash \mathbb{S}^{(i_1)}, \cdots, \mathbb{S}^{(i_2)} \backslash \mathbb{S}^{(i_2-1)}, \mathbb{S}^{(i_2+1)} \backslash \mathbb{S}^{(i_2)}, \cdots, \mathbb{S}^{(i_3)} \backslash \mathbb{S}^{(i_3-1)}, \mathbb{S}^{(i_3+1)} \backslash \mathbb{S}^{(i_3)}, \cdots
$$
until there are no merged subpaths left in the graph.

As stated before, from each $\tilde{K}$-trunk, at most $k-1$ $\phi$-paths can merge further, so at least one of the $\phi$-paths from this trunk will not go forward to merge anymore, implying there exist at most $n$ $\tilde{K}$-trunks. Note also that the number of mergings within all $K$-trunks will be upper bounded by $3Kn$, since we can only have at most $2n$ singular $K$-trunks and $n$ normal $K$-trunks. Summing up all the merged subpaths contained in all $\tilde{K}$-trunks, we conclude that the number of merged subpaths within all $\tilde{K}$-trunks is upper bounded by
$$
2C_{k-1} (L_{i_0} + L_{i_1} + \cdots + L_{i_i}+ \cdots) + nC_{k-1} (k-2) \leq 2C_{k-1} (n+(k-1)(n-1))+nC_{k-1} (k-2),
$$
(here $(k-1)(n-1)$ is the upper bound on the number of $\phi$-paths that may belong to more than one $\tilde{K}$-trunk) which implies that the number of mergings in $G$ can be upper bounded by
$$
|G|_{\mathcal{M}} \leq C_k n,
$$
for some constant $C_k$.

\end{proof}

\section{Minimum Mergings $\mathcal{M^*}$} \label{MStar}

In this section, we consider any acyclic directed graph $G$ with one source and $n$ distinct sinks. Let $M^*(G)$ denote the minimum number of mergings over all possible Menger's path sets $\alpha_i$'s, $i=1, 2, \cdots, n$, and let $\mathcal{M}^*(c_1, c_2, \cdots, c_n)$ denote the supremum of $M^*(G)$ over all possible choices of such $G$.

We also have the following ``finiteness'' theorem for $\mathcal{M}^*$:
\begin{thm} \label{main-1}
For any $c_1, c_2, \cdots, c_n$,
$$
\mathcal{M}^*(c_1, c_2, \cdots, c_n) < \infty,
$$
and furthermore, we have
$$
\mathcal{M}^*(c_1, c_2, \cdots, c_{n}) \leq \sum_{i < j} \mathcal{M}^*(c_i, c_j).
$$
\end{thm}

\begin{proof}

As illustrated in Remark~\ref{imaginary}, we extend $G$ to $\hat{G}$ by first adding $n$ imaginary sources $S_1, S_2, \cdots, S_n$, and then adding $c_i$ disjoint edges from $S_i$ to $S$ for each feasible $i$. For any such $G$ and $\hat{G}$, one checks that the original Menger's paths (from $S$ to each $R_i$ for all $i$) merge with each other fewer times than the extended Menger's paths (from $S_i$ to $R_i$ for all $i$), which implies that
$$
\mathcal{M}^*(c_1, c_2, \cdots, c_n) \leq \mathcal{M}(c_1, c_2, \cdots, c_n).
$$
The finiteness result then immediately follows from Theorem~\ref{main}. As for the inequality,
exactly the same argument of Theorem~\ref{main} applies to $\mathcal{M}^*$, thus we have for any $c_1, c_2, \cdots, c_{n+1}$
$$
\mathcal{M}^*(c_1, c_2, \cdots, c_{n}) \leq \mathcal{M}^*(c_1, c_2, \cdots, c_{n-1})+\sum_{j < n} \mathcal{M}^*(c_j, c_{n}),
$$
which implies the inequality.
\end{proof}

\begin{rem}
The same techniques as in the proof above, together with Theorem~\ref{main}, show that appropriately chosen Menger's paths merge with each other only finitely many times, if only some of the sources and/or some of the sinks are identical.
\end{rem}

\begin{rem}
Theorem~\ref{main} and Theorem~\ref{main-1} do not hold for cyclic directed graphs. As shown in Figure~\ref{cyclic}, for an arbitrary $n$, $\alpha_{2, 1}$ merges with $\alpha_{1, 2}$ at $\gamma_1, \gamma_2, \cdots, \gamma_{n-1}, \gamma_{n}$ subsequently from the bottom to the top. One checks that $\alpha_1$ and $\alpha_2$ has $n$ mergings, and there is no way to reroute $\alpha_1$ or $\alpha_2$ to decrease the number of mergings.

\begin{figure}
\psfrag{G}{$G$}
\psfrag{a1}{$\alpha_{2, 1}$} \psfrag{a2}{$\alpha_{2, 2}$}
\psfrag{b1}{$\alpha_{1, 1}$} \psfrag{b2}{$\alpha_{1, 2}$}
\psfrag{g1}{$\gamma_1$} \psfrag{g2}{$\gamma_2$} \psfrag{g3}{$\gamma_3$} \psfrag{gn1}{$\gamma_{n-1}$} \psfrag{gn}{$\gamma_n$}
\psfrag{S}{$S$}\psfrag{R1}{$R_1$} \psfrag{R2}{$R_2$}
\centerline{\includegraphics[width=3in]{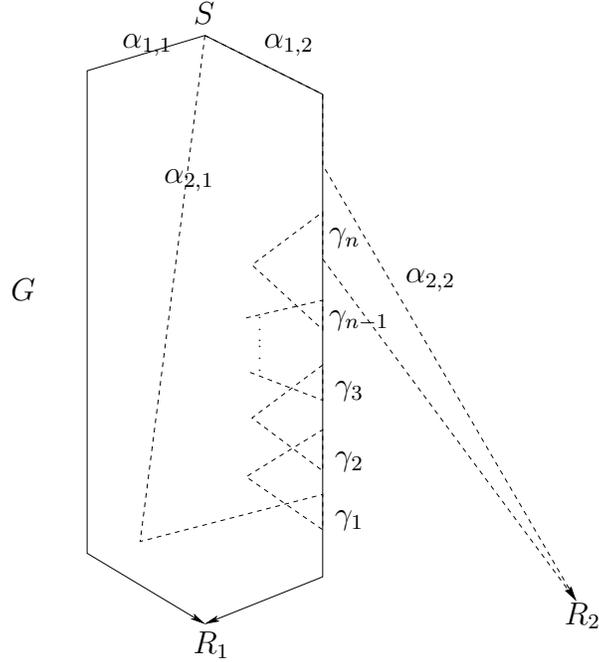}}
\caption{an counterexample}
\label{cyclic}
\end{figure}
\end{rem}

Similar to $\mathcal{M}$, $\mathcal{M}^*$ is a symmetric and ``increasing'' function.

\begin{pr}
$\mathcal{M}^*$ is symmetric on its parameters. More specifically,
$$
\mathcal{M}^*(c_1, c_2, \cdots, c_n)=\mathcal{M}^*(c_{\delta(1)}, c_{\delta(2)}, \cdots, c_{\delta(n)}),
$$
where $\delta$ is any permuation on the set of $\{1, 2, \cdots, n\}$.
\end{pr}

\begin{pr}  \label{increasing}
For $m \geq n$, $c_1 \leq c_2 \leq \cdots \leq c_n$, and $d_1 \leq d_2 \cdots \leq d_m$, if
$c_i \leq d_{m-n+i}$ for $i=1, 2, \cdots, n$, then
$$
\mathcal{M}^*(c_1, c_2, \cdots, c_n) \leq \mathcal{M}^*(d_1, d_2, \cdots, d_m).
$$
\end{pr}

\begin{pr} \label{SharedSubpath}
For $c_1 \leq c_2 \leq \cdots \leq c_n$, if $c_1+c_2+\cdots+c_{n-1} \leq c_n$, then
$$
\mathcal{M^*}(c_1, c_2, \cdots, c_n)=\mathcal{M^*}(c_1, c_2, \cdots, c_{n-1}, c_1+c_2+\cdots+c_{n-1}).
$$
\end{pr}

\begin{proof}
Given any acyclic directed graph $G$ with one source $S$ and $n$ sinks $R_1, R_2, \cdots, R_n$, where the min-cut between $S$ and $R_i$ is $c_i$, pick a set of Menger's paths $\alpha_i=\{\alpha_{i, 1}, \alpha_{i, 2}, \cdots, \alpha_{i, c_i}\}$ from $S$ to $R_i$ for all feasible $i$. If any path from $\alpha_n$, say $\beta$, does not share subpath starting from $S$ with any other paths and first merges with some path $\eta$ at merged subpath $\gamma$, then one can reroute all such $\eta$ (merging with $\beta$ at $\gamma$) by replacing $\eta[S, b(\gamma)]$ by $\beta[S, b(\gamma)]$ to reduce the merging number. Note that such possible reroutings can be done to all the paths from $\alpha_n$. As a result of such possible reroutings, at least $c_n-(c_1+c_2+\cdots+c_{n-1})$ paths from $\alpha_n$ will not merge with any paths from $\alpha_1, \alpha_2, \cdots, \alpha_{n-1}$, which implies
$$
\mathcal{M^*}(c_1, c_2, \cdots, c_n) \leq \mathcal{M^*}(c_1, c_2, \cdots, c_{n-1}, c_1+c_2+\cdots+c_{n-1}).
$$
The other direction is obvious from Proposition~\ref{increasing}. The proposition then immediately follows.
\end{proof}

\begin{pr} \label{WithOne}
For $c_1=1 \leq c_2 \leq \cdots \leq c_n$, we have
$$
\mathcal{M^*}(c_1, c_2, \cdots, c_n)=\mathcal{M^*}(c_2, \cdots, c_{n-1}, c_n).
$$
\end{pr}

\begin{proof}
Given any acyclic directed graph $G$ with one source $S$ and $n$ sinks $R_1, R_2, \cdots, R_n$, where the min-cut between $S$ and $R_i$ is $c_i$, choose Menger's paths $\alpha_2, \alpha_3, \cdots, \alpha_n$ such that the number of mergings among them is less than $\mathcal{M^*}(c_2, \cdots, c_{n-1}, c_n)$. If $\alpha_{1, 1}$ does not merge with any paths from $\alpha_2, \alpha_3, \cdots, \alpha_n$, then the number of mergings in $G$ among all $\alpha_i$'s is less than $\mathcal{M^*}(c_2, \cdots, c_{n-1}, c_n)$; if $\alpha_{1,1}$ does merge with other paths and it last merges with, say $\alpha_{i, j}$, at $\gamma$, then we can reroute $\alpha_{1, 1}$ by replacing $\alpha_{1, 1}[S, a(\gamma)]$ by $\alpha_{i, j}[S, a(\gamma)]$. With rerouted $\alpha_{1, 1}$, the number of mergings in $G$ among all $\alpha_i$'s is still less than $\mathcal{M^*}(c_2, \cdots, c_{n-1}, c_n)$, which implies
$$
\mathcal{M^*}(c_1, c_2, \cdots, c_n) \leq \mathcal{M^*}(c_2, \cdots, c_{n-1}, c_n).
$$
The other direction is obvious from Proposition~\ref{increasing}. The Proposition then immediately follows.
\end{proof}

\begin{rem}
Now we can see that in terms of the dependence on the parameters, the behaviors of $\mathcal{M}$ and $\mathcal{M}^*$ can be very different.
For instance,
\begin{itemize}
\item from Example~\ref{MWithOne}, we have $\mathcal{M}(1, 2)=2 > 1=\mathcal{M}(1, 1)$, which implies $\mathcal{M}$ does not satisfy the equality in Proposition~\ref{SharedSubpath};
\item through Proposition~\ref{SharedSubpath}, we see that
$$
\mathcal{M}^*(2c, c)=\mathcal{M}^*(c, c) \leq \mathcal{M}^*(c, c) + \mathcal{M}^*(c, c),
$$
and strict inequality in the above expression holds as long as $\mathcal{M}^*(c, c) > 0$, thus $\mathcal{M}^*$ does not satisfy the inequality in Proposition~\ref{Isolated}; namely, not like $\mathcal{M}$, $\mathcal{M}^*$ is not sup-linear in its parameters;
\item Proposition~\ref{WithOne} implies that $\mathcal{M}^*(1, n)=0$, while from Example~\ref{MWithOne}, we have $\mathcal{M}(1, n)=n$, which implies $\mathcal{M}$ does not satisfy the equality in Proposition~\ref{WithOne}.
\end{itemize}
\end{rem}

The following proposition reveals a relationship between $\mathcal{M}$ and $\mathcal{M}^*$.

\begin{pr} \label{MvsMStar} For any $n$, we have
$$
\mathcal{M}^*(n+1, n+1) \leq \mathcal{M}(n, n)-n+1.
$$
\end{pr}

\begin{proof}
Consider the case when $G$ has one source $S$ and two sinks $R_1, R_2$, and the min-cut between the source $S$ and every sink is equal to $n+1$. For each sink $R_i$, pick a set of Menger's paths $\alpha_i=\{\alpha_{i, 1}, \alpha_{i, 2}, \cdots, \alpha_{i, n+1}\}$. By Proposition~\ref{increasing}, we can assume every $\alpha_1$-path merges with certain $\alpha_2$-path and vice versa. As shown in the proof of Proposition~\ref{SharedSubpath}, we can further assume $\alpha_{1, i}$ shares subpath starting from $S$ with $\alpha_{2, i}$, $i=1, 2, \cdots, n+1$, after possible reroutings. Now, if every $\alpha_1$-path merges with some $\alpha_2$-path, for instance, $\alpha_{1, i}$ first merges with $\alpha_{2, \delta(i)}$ at merged subpath $\gamma_i$, here $\delta$ denotes certain mapping from $\{1, 2, \cdots, n+1\}$ to $\{1, 2, \cdots, n+1\}$. One then checks that there exist $i$ ($i \leq n+1$) and $m$ ($m \leq n+1$) such that $\delta^m(i)=i$, and one can further choose $m$ to be the smallest such ``period''. However, in this case certain reroutings of $\alpha_2$ can be done by replacing $\alpha_{2, \delta^j(i)}[S, b(\gamma_{\delta^{j-1}(i)})]$ by $\alpha_{1, \delta^{j-1}(i)}[S, b(\gamma_{\delta^{j-1}(i)})]$, $j=1, \cdots, m$ (here $\delta^0(i)\stackrel{\triangle}{=}i$), to reduce the merging number. So, without loss of generality, we can assume, after further possible reroutings, $\alpha_{1, n+1}$ does not merge with any other paths, and $\alpha_{2, 1}$ doesn't merge with any other paths either by similar argument; in other words, all mergings are by paths $\alpha_{1, 1}, \alpha_{1, 2}, \cdots, \alpha_{1, n}$ and paths $\alpha_{2, 2}, \alpha_{2, 3}, \cdots, \alpha_{2, n+1}$. With the fact that for $j=2, 3, \cdots, n$, $\alpha_{1, j}$ shares a subpath (which will not be counted when computing $\mathcal{M}^*(G)$) with $\alpha_{2, j}$, we establish the theorem.
\end{proof}

\begin{pr} \label{WithTwo}
For any $n$, we have
$$
\mathcal{M^*}(\underbrace{2, 2, \cdots, 2}_n)=\mathcal{M^*}(\underbrace{2, \cdots, 2, 2}_{n-1})+1.
$$
\end{pr}

\begin{proof}
Given any acyclic directed graph $G$ with one source $S$ and $n$ sinks $R_1, R_2, \cdots, R_n$, where the min-cut between $S$ and $R_i$ is $2$, pick a set of Menger's paths $\alpha_i=\{\alpha_{i, 1}, \alpha_{i, 2}\}$ from $S$ to $R_i$ for all feasible $i$. Again by a new merging, we mean a merging among $\alpha_1, \alpha_2, \cdots, \alpha_n$, however not among $\alpha_1, \alpha_2, \cdots, \alpha_{n-1}$. Assume that $\alpha_1, \alpha_2, \cdots, \alpha_{n-1}$ are chosen such that the mergings among themselves is no more than $\mathcal{M^*}(\underbrace{2, 2, \cdots, 2}_{n-1})$, we shall prove that whenever $\alpha_n$ newly merges with $\alpha_1, \alpha_2, \cdots, \alpha_{n-1}$ more than $2$ times, one can always reroute certain paths to decrease the total number of mergings within $\alpha_1, \alpha_2, \cdots, \alpha_{n}$. Apparently this will be sufficient to imply:
$$
\mathcal{M^*}(\underbrace{2, 2, \cdots, 2}_n) \leq \mathcal{M^*}(\underbrace{2, \cdots, 2, 2}_{n-1})+1.
$$

In the following, for any $j$, if we use $p$ to refer to one of the two paths in $\alpha_j$, we will use $\bar{p}$ to refer to the other path in $\alpha_j$. Consider the following two scenarios:
\begin{enumerate}
\item for two certain Menger's paths $p, q$, $p$ merges with $q$ and $\bar{p}$ merges with $\bar{q}$;
\item for a Menger's path $p \in \alpha_n$ which newly merges with $q_1, q_2, \cdots, q_l$ at subpath $\gamma$ (here we have listed all the paths merging with $p$ at $\gamma$), $p$ shares a subpath with every $q_j$ before the new merging.
\end{enumerate}

For scenario $1$, suppose $p$ merges with $q$ at $\gamma$, and $\bar{p}$ merges with $\bar{q}$ at $\eps$. Then one can
always reroute $p[S, a(\gamma)]$ using $q[S, a(\gamma)]$, reroute $\bar{p}[S, a(\eps)]$ using $\bar{q}[S, a(\eps)]$; or alternatively reroute $q[S, a(\gamma)]$ using $p[S, a(\gamma)]$, reroute $\bar{q}[S, a(\eps)]$ using $\bar{p}[S, a(\eps)]$. So in the following we assume that scenario $1$ never occurs.

For scenario $2$, suppose that before $p$ newly merges with $q_1, q_2, \cdots, q_l$ at $\gamma$, $p$ shares a subpath $\eps_j$ with every $q_j$. We can assume $\bar{p}$ merges with every $q_j[b(\eps_j), a(\gamma)]$, otherwise one can reroute $p[b(\psi_j), a(\phi)]$ using $q_j[b(\psi_j), a(\phi)]$ (and thus the new merging at $\gamma$ disappear); we can also assume for some path $i$, $\bar{q}_i$ merges with $p[b(\eps_i), a(\gamma)]$, otherwise one can reroute every $q_j[b(\eps_j), a(\gamma)]$ using $p[b(\eps_j), a(\gamma)]$ and consequently all paths $q_1, q_2, \cdots, q_l$ can be rerouted (and thus the new merging at $\gamma$ disappear). But if for some path $i$, $\bar{q}_i$ merges with $p[b(\eps_i), a(\gamma)]$, scenario $1$ occurs: $p$ merges with $q_i$, and $\bar{p}$ merges with $\bar{q}_i$. So in the following we assume scenario $2$ does not occur either, i.e., there is always some $q_i$ such that before the new merging, $p$ does not internally intersect with $q_i$.

We say $p$ newly merges with $q_i$ {\em essentially} at $\gamma$ if
\begin{enumerate}
\item before the new merging, $p$ does not internally intersect (again meaning share subpath) with $q_i$;
\item $\bar{p}$ internally intersects with $q_j[S, a(\gamma)]$;
\item $\bar{q}_i$ internally intersects with $p[S, a(\gamma)]$.
\end{enumerate}
One checks that if $p$ newly merges with some $q_i$ non-essentially at $\gamma$, then either $p[S, a(\gamma)]$ or $q_i[S, a(\gamma)]$ can be rerouted. Furthermore if $p$ newly merges with $q_i$ essentially at $\gamma$, and $\bar{p}$ last merges with $q_i[S, a(\gamma)]$ at $\eps$, then one can reroute $\bar{p}$ by replacing $\bar{p}[S, a(\eps)]$ by $\bar{q}_i[S, a(\eps)]$, so the new $\bar{p}$ shares subpath $\bar{q}_i[S, b(\eps)]$ staring from $S$;
in other words, after possible reroutings, we can further assume that $\bar{p}$ shares certain subpath with $q_i$ starting from $S$.

Now suppose $p \in \alpha_n$ newly merges twice at $\gamma_1, \gamma_2$. For $i=1, 2$, among all the Menger's paths merging with $p$ at $\gamma_i$, let $q_i$ denote an arbitrarily chosen path such that $p$ newly merges with $q_i$ at $\gamma_i$ essentially (note that $q_1 \neq q_2$ since both of them merge with $p$ essentially). If $\bar{q}_2$ merges with $p[b(\gamma_1), a(\gamma_2)]$ at subpath $\eps_1$, since $\bar{q}_2$ does not merge with $\bar{p}$ (scenario $1$ does not occur), one can reroute $p[S, a(\eps_1)]$ using $\bar{q}_2[S, a(\eps_1)]$ (then the new merging at $\gamma_1$ would disappear). Consider the case when $\bar{q}_2$ does not merge with $p[b(\gamma_1), a(\gamma_2)]$. If $\bar{q}_2$ does not merge with $q_1[S, a(\gamma_1)]$ either, one can reroute $q_2[S, a(\gamma_2)]$ using $q_1[S, a(\gamma_1)] \circ p[a(\gamma_1), a(\gamma_2)]$. Now consider the case when $\bar{q}_2$ merges with $q_1[S, a(\gamma_1)]$ and suppose $\bar{q}_2$ last merges with $q_1[S, a(\gamma_1)]$ at $\eps_2$. If $\bar{p}$ does not merge with $q_1[b(\eps_2), a(\gamma_1)]$, since $\bar{q}_2$ won't merge with $\bar{p}$, $p[S, a(\gamma_1)]$ can be rerouted using $\bar{q}_2[S, b(\eps_2)] \circ q_1[b(\eps_2), a(\gamma_1)]$ (then the new merging at $\gamma_1$ would disappear). Now consider the case when $\bar{p}$ does merge with $q_1[b(\eps_2), a(\gamma_1)]$ at subpath $\eps_3$. But in this case, one can reroute $q_2[S, a(\gamma_2)]$ using $\bar{p}[S, a(\eps_3)] \circ q_1[a(\eps_3), a(\gamma_1)] \circ p[a(\gamma_1), a(\gamma_2)]$. Apply the arguments above to arbitrarily chosen pair $q_1, q_2$ essentially merging with $p$, together with the fact that non-essential merging will disappear after appropriate reroutings, we conclude that ultimately certain reroutings to reduce the number of mergings are always possible when $p \in \alpha_n$ newly merges twice.

Now suppose $p \in \alpha_n$ newly merges at $\gamma_1$, and $\bar{p} \in \alpha_n$ newly merges at $\gamma_2$. Let $q_1$ denote an arbitrarily chosen path, among all the paths merging with $p$ at $\gamma_1$, such that $p$ newly merges with $q_1$ at $\gamma_1$ essentially; let $q_2$ denote an arbitrarily chosen path, among all the paths merging with $\bar{p}$ at $\gamma_2$, such that $\bar{p}$ newly merges with $q_2$ at $\gamma_2$ essentially (again one checks that $q_1 \neq q_2$ since they essentially merge with $p, \bar{p}$, respectively). Apparently $q_1, q_2$ must merge with each other, otherwise one can reroute $p[S, a(\gamma_1)]$ using $q_1[S, a(\gamma_1)]$ and reroute $\bar{p}[S, a(\gamma_2)]$ using $q_2[S, a(\gamma_2)]$ (then the two new mergings would disappear). Suppose $q_1$ and $q_2$ last merge at $\eps_1$. We claim that $\bar{p}$ must merge with $q_1[b(\eps_1), a(\gamma_1)]$, otherwise one can reroute $p[S, a(\gamma_1)]$ using $q_2[S, b(\eps_1)] \circ q_1[b(\eps_1), a(\gamma_1)]$ ($p$ shares subpath with $q_2$ from $S$ and does not merge with $q_1$ before $\gamma_1$). Furthermore $\bar{p}$ must merge with $q_1[b(\eps_1), a(\gamma_1)]$ at least once before $a(\gamma_2)$ (in other words, $\bar{p}[S, a(\gamma_2)]$ must merge with $q_1[b(\eps_1), a(\gamma_1)]$), since otherwise, say $\bar{p}[b(\gamma_2), R_n]$ merges with $q_1[b(\eps_1), a(\gamma_1)]$ at $\eps_2$, then one can reroute $\bar{p}[S, a(\eps_2)]$ with $q_1[S, a(\eps_2)]$ (thus the new merging at $\gamma_2$ would disappear). Similarly $p[S, a(\gamma_1)]$ must merge with $q_2[b(\eps_1), a(\gamma_2)]$. Now suppose $\bar{p}[S, a(\gamma_2)]$ first merges with $q_1[b(\eps_1), a(\gamma_1)]$ at subpath $\eps_2$. Since scenario $1$ does not occur, $\bar{q}_1$ won't merge with $p$, therefore it must share certain subpath with $p$ staring from $S$ (here we remind the reader that $p$ newly merges with $q_1$ essentially, so $\bar{q}_1$ will either merge with or share certain subpath with $p$ from $S$). Similarly suppose $p[S, a(\gamma_1)]$ first merges with $q_2[b(\eps_1), a(\gamma_2)]$ at $\eps_3$, then $\bar{q}_2$ must share certain subpath with $\bar{p}$ staring from $S$. Now since scenario $1$ does not occur, either $\bar{q}_2$ won't merge with $q_1[b(\eps_1), a(\eps_2)]$ or $\bar{q}_1$ won't merge with $q_2[b(\eps_1), a(\eps_3)]$. If $\bar{q}_2$ does not merge with $q_1[b(\eps_1), a(\eps_2)]$, then one can reroute $q_2[b(\eps_1), a(\gamma_2)]$ with $q_1[b(\eps_1), a(\eps_2)] \circ \bar{p}[a(\eps_2), a(\gamma_2)]$; if $\bar{q}_1$ does not merge with $q_2[b(\eps_1), a(\eps_3)]$, then one can reroute $q_1[b(\eps_1), a(\gamma_1)]$ with $q_2[b(\eps_1), a(\eps_3)] \circ p[a(\eps_3), a(\gamma_1)]$.
Apply the arguments above to arbitrarily chosen pair $q_1, q_2$ essentially merging with $p$, together with the fact that non-essential merging will disappear after appropriate reroutings, we conclude that ultimately certain reroutings to reduce the number of mergings are always possible when when $p \in \alpha_n$ newly merges and $\bar{p} \in \alpha_n$ newly merges.

For the other direction, assume that the subgraph consisting of $\alpha_1, \alpha_2, \cdots, \alpha_{n-1}$ achieves $\mathcal{M^*}(\underbrace{2, 2, \cdots, 2}_{n-1})$, we add $\alpha_n$ such that for $i=1, 2$, $\alpha_{n, i}$ share subpath with $\alpha_{n-1, i}$, $\alpha_n$ only merges with $\alpha_{n-1}$ once, say $\alpha_{n, 1}$ merges with $\alpha_{n-1, 2}$ at $\gamma$, where $\gamma$ is a largest merged subpath. One checks the graph consisting $\alpha_1, \alpha_2, \cdots, \alpha_n$ has $\mathcal{M^*}(\underbrace{2, \cdots, 2, 2}_{n-1})+1$ mergings, and the number of mergings can't be reduced, implying
$$
\mathcal{M^*}(\underbrace{2, 2, \cdots, 2}_n) \geq \mathcal{M^*}(\underbrace{2, \cdots, 2, 2}_{n-1})+1.
$$
We thus prove the proposition.

\end{proof}

\begin{exmp}
It immediately follows from Proposition~\ref{WithOne} that
$$
\mathcal{M^*}(1, 1, \cdots, 1)=0.
$$
\end{exmp}

\begin{exmp} \label{NotEqual}
It immediately follows from Proposition~\ref{WithTwo} that
$$
\mathcal{M^*}(\underbrace{2, 2, \cdots, 2}_n)=n-1,
$$
which was first shown in~\cite{Fr2006}. In particular, $\mathcal{M}(2, 2)=1$. Further together with Proposition~\ref{SharedSubpath}, we have $\mathcal{M}^*(2, m)=1$ for $m \geq 2$. Note that
$$
\mathcal{M^*}(\underbrace{2, 2, \cdots, 2}_n) < \sum_{1 \leq i < j \leq n} \mathcal{M}^*(2, 2),
$$
which implies the inequality in Theorem~\ref{main-1} may not hold for certain cases.
\end{exmp}

\begin{exmp}
It follows from Proposition~\ref{MvsMStar} that
$$
\mathcal{M}^*(3, 3) \leq \mathcal{M}(2, 2)-1=4.
$$
One checks that the graph depicted by Figure~\ref{threethree} does not allow any rerouting to reduce the number of mergings, which implies $\mathcal{M^*}(3,3)=4$. Applying Theorem~\ref{main-1}, we have
$$
\mathcal{M}^*(\underbrace{3, 3, \cdots, 3}_n) \leq 2n(n-1).
$$

\begin{figure}
\psfrag{a11}{$\alpha_{1, 1}$} \psfrag{a12}{$\alpha_{1, 2}$} \psfrag{a13}{$\alpha_{1, 3}$}
\psfrag{a21}{$\alpha_{2, 1}$} \psfrag{a22}{$\alpha_{2, 2}$} \psfrag{a23}{$\alpha_{2, 3}$}
\psfrag{S}{$S$}\psfrag{R1}{$R_1$}\psfrag{R2}{$R_2$}
\centerline{\includegraphics[width=2.5in]{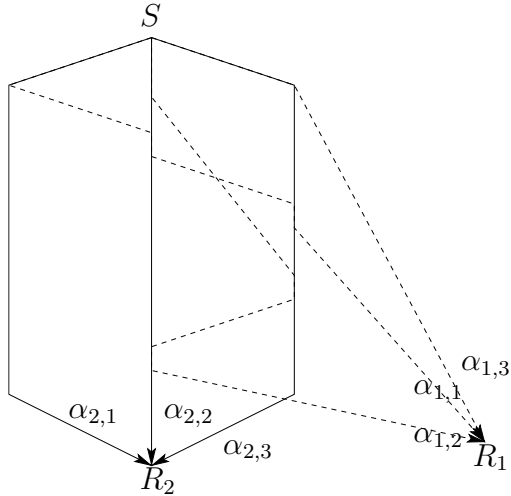}}
\caption{an example achieving $\mathcal{M}^*(3, 3)$}
\label{threethree}
\end{figure}
\end{exmp}

\section{Motivations}

Mergings in directed graphs naturally relate to ``congestions'' of traffic flows in various networks. Particularly, in network coding theory~\cite{Yeung2005}, which studies digital communication networks carrying information flow~\cite{Ahlswede2000}, computations and estimations of $\mathcal{M}$ and $\mathcal{M}^*$ have drawn much interest recently. Recent related work in network coding theory listed in this section are done in very different languages; we shall briefly introduce network coding theory and describe these work using the terminology and notations in this paper.

Network coding is a novel technique to improve the capability of networks (directed graphs) to transfer digital information between senders (sources) and receivers (sinks). Before network coding, information is transferred among networks using the traditional routing scheme, where intermediate nodes (vertices) can only forward and duplicate the received information. In contrast to the routing scheme, the idea of network coding is to allow intermediate nodes to ``combine'' data received from different incoming links (edges), eventually boosting the transmission rate of the network.

For a very comprehensive introduction to network coding theory, we refer to~\cite{Yeung2005}. Here, we roughly illustrate the idea of network coding using the following famous ``butterfly network''~\cite{Li2003}. Consider the network depicted in Figure~\ref{Butterfly}, where each link has capacity $1$ bit per time unit and there is no processing delay at each node. Two binary bits $a, b$ are to be transmitted from the source $S$ to $Y$ and $Z$. If we ignore the transmission to $Z$, we can use path $S \to T \to Y$ to transmit $a$, and use path $S \to U \to W \to X \to Y$ to transmit $b$ simultaneously; similarly ignoring the transmission to $Y$, we can use path $S \to U \to Z$ to transmit $a$, and use path $S \to T \to W \to X \to Z$ to transmit $b$ simultaneously. Note that paths $S \to U \to W \to X \to Y$ and $S \to T \to W \to X \to Z$ merge at $W \to X$. If the traditional routing scheme is assumed, $W \to X$ will become a ``bottleneck'' for simultaneous data transmission to $Y$ and $Z$, since for each time unit $W \to X$ can either carry $a$ or $b$, but not both at once. Thus under the routing scheme, completion of data transmission takes at least $2$ time units. Allowing intermediate nodes to recode the data from the incoming links, network coding scheme will provide a solution to speed up the data transmission: the ``bottleneck'' $W \to X$ carry $a$ and $b$ at the same time by carrying $a+b$, here $+$ denotes the exclusive-OR on $a, b$. Then as shown in Figure~\ref{Butterfly}(b), $Y$ will receive $a$ and $a+b$, from which $b$ can be decoded; at the same time unit $Z$ will receive $b$ and $a+b$, from which $a$ can be decoded. In other words, with the encoding at node $W$, $Y$ and $Z$ can receive the complete data simultaneously within $1$ time unit.
\begin{figure}
\centerline{\includegraphics[width=4in]{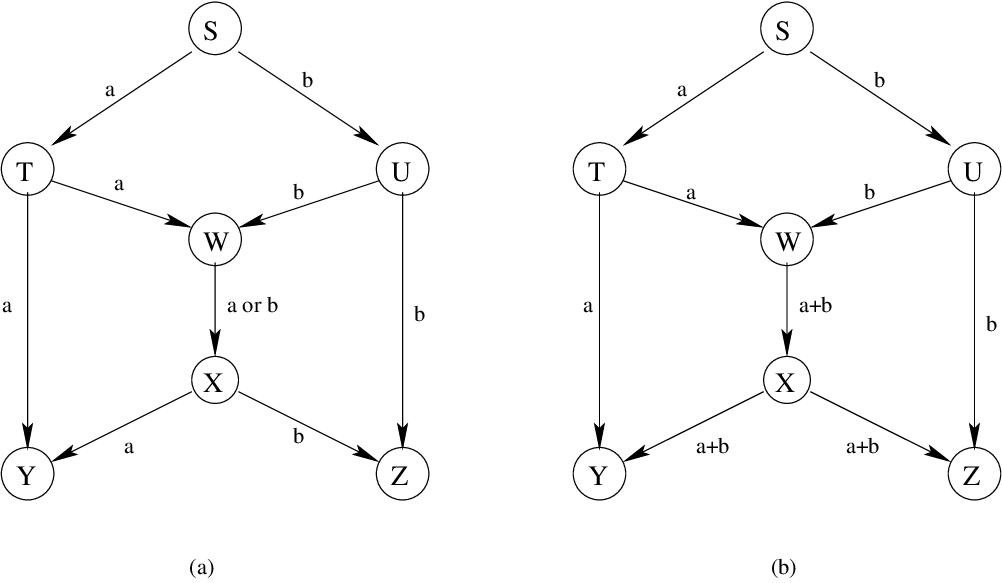}}
\caption{network coding on the Butterfly network}
\label{Butterfly}
\end{figure}

Now consider a general network with one sender $S$ and $n$ receivers $R_1, R_2, \cdots, R_n$, where each edge has capacity $1$ bit per time unit and there is no processing delay at each node. Suppose that each $R_i$ has the same min-cut $c$ with the sender $S$, and $c$ bit information are to be transmitted from $S$ to all $R_i$'s. Ignoring the presence of other receivers, any set of Menger's paths from $S$ to a receiver is able to carry data to the receiver at the maximum possible rate $c$; however for simultaneous data transmission, any merging among these Menger's paths will become a bottleneck. It has been shown~\cite{Ahlswede2000, Li2003} that with appropriate network coding at the merging nodes, all the receivers can receive the information at the maximum possible rate $c$.

In a network coding scheme, we call a node an ``encoding node'' if this node recodes the data from the incoming links, rather than simply duplicating and forwarding the incoming date. It is important to minimize, for a given network, the number of nodes which are needed to be equipped with such encoding capabilities, since these nodes are typically more expensive than other forwarding nodes, and may increase the overall complexity of the network. Since for given sets of Menger's paths from the source to the receivers, encoding operations are only needed at merging nodes among these paths, $\mathcal{M}$ and $\mathcal{M}^*$ with appropriate parameters will naturally give upper bounds on the number of necessary encoding nodes for a given network. In particular, for an acyclic network $G$ with one source and multiple sinks, as suggested by Lemma $13$ of~\cite{la06}, the minimum number of coding operations (required to guarantee all receivers receive data at the maximum possible rate) is equal to $M^*(G)$.

It was first conjectured that $\mathcal{M}(c_1, c_2, \cdots, c_n)$ is finite in~\cite{Ta2003}. More specifically the authors proved that (see Lemma $10$ of~\cite{Ta2003}) if $\mathcal{M}(c_1, c_2)$ is finite for all $c_1, c_2$, then
$\mathcal{M}(c_1, c_2, \cdots, c_n)$ is finite as well. To support the conjecture, the authors showed that
$\mathcal{M}(2, c)$ is finite for any $c$, and subsequently $\mathcal{M}(\underbrace{2, 2, \cdots, 2}_n, c)$ is finite for any $n$ and $c$. Lemma~\ref{twotwo} shows that indeed $\mathcal{M}(c_1, c_2)$ is finite for all $c_1, c_2$, thus the conjecture is true.

As for $\mathcal{M}^*$, the authors of~\cite{Fr2006} use the idea of ``subtree decomposition'' to first prove that
$$
\mathcal{M}^*(\underbrace{2, 2, \cdots, 2}_n)=n-1.
$$
Although their idea seems to be difficult to generalize to other parameters, it does allow us to gain deeper understanding about the topological structure of minimum mergings achieving graph for this special case.
It was first shown in~\cite{la06} that $\mathcal{M}^*(c_1, c_2)$ is finite for all $c_1, c_2$ (see Theorem $22$ of~\cite{la06}), and subsequently $\mathcal{M}^*(c_1, c_2, \cdots, c_n)$ is finite all $c_1, c_2, \cdots, c_n$.
The proof of Lemma~\ref{twotwo} is inspired by and follows closely the spirit of the proof of Theorem $22$ of~\cite{la06}. One of the differences between the approach in~\cite{la06} and ours is that we start with arbitrarily chosen Menger's paths, and focus on transformations (more specifically, merging number reducing reroutings) of these paths, which allow us to see how $\mathcal{M}, \mathcal{M}^*$ depend on the min-cuts.

\medskip \textbf{Acknowledgements:} We are very grateful to Raymond Yeung and Sidharth Jaggi, who have pointed out the related work~\cite{Fr2006, la06, Ta2003} in network coding theory. We also thank Wenan Zang and Sheng Huang for pointing out some mistakes in an earlier version of this manuscript.

\end{document}